\documentclass[a4paper,11pt]{article}

\usepackage[mathlines]{lineno}
\usepackage{amsmath}
\usepackage{amssymb}
\usepackage{amsthm}
\usepackage{graphicx}
\usepackage{enumerate}
\usepackage{cite}
\usepackage{fullpage}
\usepackage{scalerel}
\usepackage[usenames]{color}
\usepackage[nodayofweek]{datetime}
\usepackage{hyperref}

\allowdisplaybreaks

\newcommand{\Real}{\ensuremath{\mathbb{R}}}
\newcommand{\Plane}{\ensuremath{\mathbb{R}^2}}

\newsavebox\Tbox
\savebox\Tbox{
  \unitlength=0.1ex%
  \begin{picture}(8,10)
    \linethickness{0.1mm}
    \multiput(0,0)(0,2){4}{\line(0,1){1}}
    \multiput(0,0)(2,0){4}{\line(1,0){1}}
    \multiput(7,7)(0,-2){4}{\line(0,-1){1}}
    \linethickness{0.2mm}
    \put(7,7){\line(-1,0){7}}
    \put(3.5,7){\circle*{3.5}}
  \end{picture}}
\newsavebox\Bbox
\savebox\Bbox{
  \unitlength=0.1ex%
  \begin{picture}(8,10)
    \linethickness{0.1mm}
    \multiput(7,7)(0,-2){4}{\line(0,-1){1}}
    \multiput(0,0)(0,2){4}{\line(0,1){1}}
    \multiput(7,7)(-2,0){4}{\line(-1,0){1}}
    \linethickness{0.2mm}
    \put(0,0){\line(1,0){7}}
    \put(3.5,0){\circle*{3.5}}
  \end{picture}}
\newsavebox\Lbox
\savebox\Lbox{
  \unitlength=0.1ex%
  \begin{picture}(8,10)
    \linethickness{0.1mm}
    \multiput(0,0)(2,0){4}{\line(1,0){1}}
    \multiput(7,7)(0,-2){4}{\line(0,-1){1}}
    \multiput(7,7)(-2,0){4}{\line(-1,0){1}}
    \linethickness{0.2mm}
    \put(0,0){\line(0,1){7}}
    \put(0,3.5){\circle*{3.5}}
  \end{picture}}
\newsavebox\Rbox
\savebox\Rbox{
  \unitlength=0.1ex%
  \begin{picture}(8,10)
    \linethickness{0.1mm}
    \multiput(0,0)(0,2){4}{\line(0,1){1}}
    \multiput(0,0)(2,0){4}{\line(1,0){1}}
    \multiput(7,7)(-2,0){4}{\line(-1,0){1}}
    \linethickness{0.2mm}
    \put(7,7){\line(0,-1){7}}
    \put(7,3.5){\circle*{3.5}}
  \end{picture}}
\newsavebox\Abox
\savebox\Abox{
  \unitlength=0.1ex%
  \begin{picture}(8,10)
    \linethickness{0.1mm}
    \multiput(0,0)(0,2){4}{\line(0,1){1}}
    \multiput(0,0)(2,0){4}{\line(1,0){1}}
    \multiput(7,7)(0,-2){4}{\line(0,-1){1}}
    \multiput(7,7)(-2,0){4}{\line(-1,0){1}}
  \end{picture}}

\DeclareRobustCommand\mt{{\scalerel*{{\mathord{\usebox{\Tbox}}}}{b}}}
\DeclareRobustCommand\mb{{\scalerel*{{\mathord{\usebox{\Bbox}}}}{b}}}
\DeclareRobustCommand\ml{{\scalerel*{{\mathord{\usebox{\Lbox}}}}{b}}}
\DeclareRobustCommand\mr{{\scalerel*{{\mathord{\usebox{\Rbox}}}}{b}}}
\DeclareRobustCommand\ma{{\scalerel*{{\mathord{\usebox{\Abox}}}}{b}}}

\newcommand{\VD}{\ensuremath{\mathsf{VD}}}
\newcommand{\VG}{\ensuremath{V\!G}}
\newcommand{\GSq}{\ensuremath{\mathfrak{G}}}
\newcommand{\VR}{\ensuremath{{V\!R}}}
\newcommand{\OCH}{\ensuremath{\mathsf{OH}}}
\newcommand{\MES}{\ensuremath{\mathfrak{M}}}
\newcommand{\bMES}{\ensuremath{\mathfrak{M}^\square}}
\newcommand{\Mes}{\ensuremath{M}}

\newcommand{\OS}{\ensuremath{\mathbb{O}}}
\newcommand{\ct}{\ensuremath{\kappa}}
\newcommand{\Vin}{\ensuremath{V_\mathrm{in}}}

\newcommand{\Sq}{\ensuremath{\mathfrak{S}}}
\newcommand{\cone}{\ensuremath{\angle}}

\newcommand{\Q}{\ensuremath{\mathcal{Q}}}
\newcommand{\T}{\ensuremath{\mathcal{T}}}

\newtheoremstyle{mytheorem}{\topsep}{\topsep}{\slshape}{}{\bfseries}{}{.5em}{}
\theoremstyle{mytheorem}
\newtheorem{theorem}{Theorem}[section]
\newtheorem{lemma}[theorem]{Lemma}

\newtheorem{corollary}[theorem]{Corollary}

\theoremstyle{definition}

\theoremstyle{remark}

\newtheorem*{remark_nonumber}{Remark}

\makeatletter
\newbox\ProofSym
\setbox\ProofSym=\hbox{%
\unitlength=0.18ex%
\begin{picture}(10,10)
\put(0,0){\framebox(9,9){}} \put(0,3){\framebox(6,6){}}
\end{picture}}
\renewenvironment{proof}[1][Proof.]{\O@proof{#1}}{\O@endproof}
\def\O@proof#1{\trivlist
   \@topsep\z@\@topsepadd\smallskipamount%
   \@ifstar{\item[]}{\item[\hskip\labelsep\it #1 ]}}
\def\O@endproof{\hfill\copy\ProofSym\linebreak[3mm]\endtrivlist}
\makeatother

\def\denseitems{
    \itemsep1pt plus1pt minus1pt
    \parsep0pt plus0pt
    \parskip0pt\topsep0pt}
\newcommand*\patchAmsMathEnvironmentForLineno[1]{%
  \expandafter\let\csname old#1\expandafter\endcsname\csname #1\endcsname
  \expandafter\let\csname oldend#1\expandafter\endcsname\csname end#1\endcsname
  \renewenvironment{#1}%
     {\linenomath\csname old#1\endcsname}%
     {\csname oldend#1\endcsname\endlinenomath}}%
\newcommand*\patchBothAmsMathEnvironmentsForLineno[1]{%
  \patchAmsMathEnvironmentForLineno{#1}%
  \patchAmsMathEnvironmentForLineno{#1*}}%
\AtBeginDocument{%
\patchBothAmsMathEnvironmentsForLineno{equation}%
\patchBothAmsMathEnvironmentsForLineno{align}%
\patchBothAmsMathEnvironmentsForLineno{flalign}%
\patchBothAmsMathEnvironmentsForLineno{alignat}%
\patchBothAmsMathEnvironmentsForLineno{gather}%
\patchBothAmsMathEnvironmentsForLineno{multline}%
}

\begin{document}


\title{Empty Squares in Arbitrary Orientation Among Points%
\thanks{%
This work was supported by Basic Science Research Program through the National
Research Foundation of Korea (NRF) funded
by the Ministry of Education (2018R1D1A1B07042755).
}
}

\author{%
Sang Won Bae\footnote{%
Division of Computer Science and Engineering, Kyonggi University, Suwon, Korea.
Email: \texttt{swbae@kgu.ac.kr} }
\and
Sang Duk Yoon\footnote{%
Department of Service and Design Engineering, Sungshin Women's University, Seoul, Korea.
Email: \texttt{sangduk.yoon@sungshin.ac.kr} }
}

\date{%
\today\quad\currenttime
}

\maketitle

\begin{abstract}
This paper studies empty squares in arbitrary orientation
among a set $P$ of $n$ points in the plane.
We prove that the number of empty squares with four contact pairs
is between $\Omega(n)$ and $O(n^2)$, and that these bounds are tight,
provided $P$ is in a certain general position.
A contact pair of a square is a pair of a point $p\in P$ and
a side $\ell$ of the square with $p\in \ell$.
The upper bound $O(n^2)$ also applies to the number of empty squares
with four contact points,
while we construct a point set among which there is no square
of four contact points.
These combinatorial results are based on new observations on
the $L_\infty$ Voronoi diagram with the axes rotated and
its close connection to empty squares in arbitrary orientation.
We then present an algorithm that maintains
a combinatorial structure of the $L_\infty$ Voronoi diagram of $P$,
while the axes of the plane continuously rotates by $90$ degrees,
and simultaneously reports all empty squares with four contact pairs among $P$
in an output-sensitive way within $O(s\log n)$ time and $O(n)$ space,
where $s$ denotes the number of reported squares.
Several new algorithmic results are also obtained:
a largest empty square among $P$ and
a square annulus of minimum width or minimum area that encloses $P$
over all orientations
can be computed in worst-case $O(n^2 \log n)$ time.\\

\end{abstract}

\section{Introduction} \label{sec:intro}

We start by posing the following combinatorial question:
\begin{quote}
 \textit{%
 Given a set $P$ of $n$ points in a proper general position in $\Plane$
 how many empty squares in arbitrary orientation
 whose boundary contains four points in $P$ can there be?
 }
\end{quote}
For a square, its contact pair is a pair of a point $p\in P$
and a side $\ell$ of it such that $p \in \ell$,
regarding $\ell$ as a segment including its endpoints.
An analogous question asks the number of empty squares with
four contact pairs.
These questions can be seen as a new variant of
the Erd\H{o}s--Szekeres problem~\cite{es-cpg-35,es-sepeg-61} for empty squares.
Let $s=s(P)$ and $s^*=s^*(P)$ be the number of
empty squares with four contact points
and with four contact pairs, respectively.
In this paper, we prove that
$0 \leq s < c_1 n^2$ and $c_2 n < s^* < c_3 n^2$
for some constants $c_1,c_2,c_3>0$.
These lower and upper bounds are tight by the existence of
point sets with the asymptotically same number of such squares.
These questions and results are shown to be intrinsic to
several computational problems on empty squares,
implying new algorithmic results.

For the purpose, we provide a solid understanding
of empty squares in arbitrary orientation among $P$
by establishing a geometric and topological relation among those squares.
The family of axis-parallel empty squares is well understood
by the $L_\infty$ Voronoi diagram of $P$.
For example, each empty axis-parallel square with three points in $P$
on the boundary defines a vertex of the diagram and
the locus of the centers of empty axis-parallel squares
with two common points in $P$ on the boundary determines an edge.
We extend this knowledge to those in arbitrary orientation
by investigating the $L_\infty$ Voronoi diagrams of $P$
with the axes rotated.
Our proof for the above combinatorial problems is based on
new observations made upon the consideration of the Voronoi diagram
in this way.
This also motivates the problem of maintaining the $L_\infty$ diagram of $P$
while the axes continuously rotates.
A noteworthy observation states that
every combinatorial change of the diagram during the rotation of the axes
corresponds to an empty square with four contact pairs.
Hence, the total amount of changes of the diagram
is bounded by $\Theta(s^*)= O(n^2)$.

Based on the observations,
we then present an output-sensitive algorithm that finds
all empty squares with four contact pairs in $O(s^* \log n)$ time
using $O(n)$ space.
Our algorithm indeed maintains a combinatorial description of
the $L_\infty$ Voronoi diagram
as the axes continuously rotates by $90$ degrees
by capturing every occurrence of such special squares
and handling them so that the diagram is correctly maintained as the invariants.
To our best knowledge, there was no such algorithmic result in the literature.
Our algorithm also applies a couple of other geometric problems,
achieving new algorithmic results:
\begin{enumerate}[(1)] \denseitems
 \item A \emph{largest empty square} in arbitrary orientations
 can be found in $O(n^2 \log n)$ time.
 This improves the previous $O(n^3)$-time algorithm by
 Bae~\cite{b-marsap-19X}.
 We also solve some query versions of this problem.
 \item A \emph{square annulus} in arbitrary orientation
 with minimum width or area can be computed in $O(n^2\log n)$ time,
 improving the previous $O(n^3\log n)$ and $O(n^3)$-time
 algorithms~\cite{b-marsap-19X,b-cmwsaao-18}.
\end{enumerate}

\paragraph*{Related work.}

Our combinatorial problem on the number of empty squares in arbitrary orientation
indeed initiates a new variant of the Erd\H{o}s--Szekeres problem.
In its original version, also known as the \emph{Happy End Problem},
Erd\H{o}s and Szekeres ask the minimum number $N(k)$ for any $k\geq 3$ such that
any set of $N(k)$ points in general position in $\Plane$
contains $k$ points in convex position~\cite{es-cpg-35,es-sepeg-61}.
It is known that $N(4)=5$, $N(5)=9$~\cite{es-cpg-35},
$N(6)=17$~\cite{sp-cs17pesp-06},
and $N(k) \geq 2^{k-2}+1$ for any $k\geq 3$~\cite{es-sepeg-61}.
The exact value of $N(k)$ for any $k>7$ is not known, while
it is strongly believed that $N(k) = 2^{k-2}+1$.
Erd\H{o}s also posed an \emph{empty} version of the problem
which asks the minimum number $H(k)$ for any $k\geq 3$ such that
any set of $H(k)$ points in general position contains
$k$ points that form an empty convex $k$-gon~\cite{e-sepeg-78}.
It is trivial to see that $H(3)=3$ and $H(4)=5$.
Harborth showed that $H(5)=10$~\cite{h-kfep-78},
while Horton~\cite{h-snec7g-83} proved $H(k)$ is unbounded for any $k\geq 7$.
The finiteness of $H(6)$ was proven by
Nicol\'{a}s~\cite{n-eht-07} and Gerken~\cite{g-echpps-08},
independently.

Many other variants and generalizations on the Erd\H{o}s--Szekeres problem
have been studied in the literature.
For more details on this subject,
see survey papers by Morris and Soltan~\cite{ms-esppcps-00} and~\cite{ms-esp-16}.
Among them the most relevant to us is the problem of bounding
the number of empty convex $k$-gons whose corners are chosen from $P$,
called \emph{$k$-holes}.
The maximum number of $k$-holes among $n$ points is proven
to be $\Theta(n^k)$ for any $k\geq 3$ and any sufficiently large $n$
by B\'{a}r\'{a}ny and Valtr~\cite{bv-pfest-98}.
Let $X_k(n)$ be the minimum number of $k$-holes among $P$
over all point sets $P$ of $n$ point in general position in $\Plane$.
Note that $X_k(n) = 0$ for any $k\geq 7$ and any $n\geq 1$
by Horton~\cite{h-snec7g-83}.
B\'{a}r\'{a}ny and F\"{u}redi~\cite{bf-eses-87} proved
that $X_k(n) = \Theta(n^2)$ for $k=3$ and $4$
and, $\Omega(n) \leq X_k(n) \leq O(n^2)$ for $k=5$ and $6$.
There has been constant effort to narrow the gap of constant factors
hidden above by several researchers;
see B\'{a}r\'{a}ny and Valtr~\cite{bv-mneppps-95},
Valtr~\cite{v-mneppps-95}, Dumistrescu~\cite{d-psfecp-00},
Pinchasi, Radoi\v{c}i\'{c}, and Sharir~\cite{prs-ecppps-06},
and references therein.

Since any four points in $P$ lying on the boundary of an empty square
form a $4$-hole,
$s$ cannot exceed the number of $4$-holes among $P$.
This, however, gives only trivial upper and lower bounds on $s$ and $s^*$.
In this paper, we give asymptotically tight bounds on $s$ and $s^*$
since we rather focus on algorithmic applications of the bounds.

A lot of algorithmic results on convex $k$-gons and $k$-holes are
also known by researchers.
Dobkin, Edelsbrunner, and Overmars~\cite{deo-secp-90}
presented an algorithm that enumerates all convex $k$-holes for $3\leq k \leq 6$.
Their algorithm in particular for $k=3,4$ is indeed output-sensitive
in time proportional to the number of reported $k$-holes.
Rote et al.~\cite{rwwz-cksckgp-91} presented an $O(n^{k-2})$-time
algorithm that exactly counts the number of convex $k$-gons
and soon improved to $O(n^{\lceil k/2\rceil})$ by
Rote and Woeginger~\cite{rw-cckgpps-92}.
Later, Mitchell et al.~\cite{mrsw-ccppps-95} presented
an $O(n^3)$-time algorithm that counts the number of convex $k$-gons
and convex $k$-holes for any $k\geq 4$.
Eppstein et al.~\cite{eorw-fmakg-92} considered
the problem of finding a convex $k$-gon or $k$-hole with minimum area
and showed how to solve it in $O(n^2)$ time for $k=3$ and $O(kn^3)$ time
for $k\geq 4$,
which soon be improved to $O(n^2\log n)$ for constant $k\geq 4$
by Eppstein~\cite{e-namakg-92}.
Boyce et al.~\cite{bddg-fep-85} considered a maximization problem
that finds a maximum-area or perimeter convex $k$-gon
and obtained an $O(kn\log n + n\log^2 n)$-time algorithm,
which has bee improved to $O(kn + n\log n)$ by
Aggarwal et al.~\cite{akmsw-gamsa-87}.
Drysdale and Jaromcyz~\cite{dj-nlbmampkg-89} showed
an $\Omega(n\log n)$ for this problem.

The problem of maintaining the $L_\infty$ (or, equivalently, $L_1$)
Voronoi diagram while the axes rotates cannot be found in the literature.
A similar paradigm about rotating axes can be seen
in Bae et al.~\cite{blacc-cmarchls-09} in which
the authors show how to maintain the \emph{orthogonal convex hull} of $P$
in $O(n^2)$ time.
Alegr\'{i}­a-Galicia et al.~\cite{aosu-obhpps-18} recently improved it
into $O(n\log n)$ time using $O(n)$ space.
As will be seen in the following, the orthogonal convex hull of $P$ indeed
describes the unbounded edges of the $L_\infty$ Voronoi diagram of $P$.

The problem of finding a largest empty square is a square variant of
the well-known \emph{largest empty rectangle problem}.
The largest empty rectangle problem is one of the most intensively studied
problems in early time of computational geometry.
Its original version asks to find an axis-parallel empty rectangle of
maximum area among $P$.
There are two known general approaches to solve the problem:
one approach enumerates and checks all \emph{maximal empty rectangles}
and the other does not.
Naamad et al.~\cite{nlh-merp-84} presented an algorithm for the largest empty rectangle problem that runs in $O(\min\{r \log n, n^2\})$ time,
where $r$ denotes the number of maximal empty rectangles.
In the same paper, it is also shown that $r = O(n^2)$ in the worst case
and $r = O(n \log n)$ in expectation.
This algorithm was improved by Orlowski~\cite{o-nalerp-90}
to $O(n\log n + r)$ time.
A divide-and-conquer $O(n\log^3 n)$-time algorithm
without enumerating all maximal empty rectangles
has been presented by Chazelle et al.~\cite{cdl-cler-86}.
This was later improved to $O(n\log^2 n)$ worst-case time by
Aggarwal and Suri~\cite{as-facler-87}.
Mckenna et al.~\cite{mrs-flrop-85} proved a lower bound of $\Omega(n \log n)$
for this problem.
It is still open whether the problem can be solved in $O(n \log n)$ worst-case time.
Augustine et al.~\cite{admnrs-qlegodl-10} and Kaplan et al.~\cite{kmns-smqmmpmmta-17} considered a query version of the problem.

Chauduri et al.~\cite{cnd-lerps-03} considered the problem with
rectangles in arbitrary orientation.
They showed that there are $O(n^3)$ combinatorially different classes
of maximal empty rectangles over all orientations,
and presented a cubic-time algorithm that enumerates all those classes.
Hence, a largest empty rectangle among $P$ in arbitrary orientation
can be computed in $O(n^3)$ time.

The largest empty square problem, however, has attained relatively less
interest.
It is obvious that a largest empty axis-parallel square
can be found in $O(n\log n)$ time by computing
the $L_\infty$ Voronoi diagram of $P$~\cite{l-tdvdlpm-80,lw-vdl1m2dsa-80},
as also remarked in early papers, including
Naamad et al.~\cite{nlh-merp-84} and Chazelle et al.~\cite{cdl-cler-86}.
More precisely, all maximal empty axis-parallel squares among $P$ are described by
the vertices and edges of the $L_\infty$ Voronoi diagram of $P$.
It is rather surprising, however, that no further result about
the largest empty square problem in arbitrary orientation
has been come up with for over three decades, to our best knowledge.
From an easy observation that any maximal empty square in arbitrary orientation
is contained in a maximal empty rectangle,
Bae~\cite{b-marsap-19X} recently showed how to compute a largest empty square
in arbitrary orientation in $O(n^3)$ time.
In this paper, we improve this to $O(n^2\log n)$ time.

This also applies to the \emph{square annulus problem} in arbitrary orientation.
A square annulus is the closed region between a pair of concentric squares
in a common orientation.
The square annulus problem asks to find a square annulus of
minimum width or minimum area that encloses a given set $P$ of points.
Abellanas et al.~\cite{ahimpr-bfr-03} and
Gluchshenko et al.~\cite{ght-oafepramw-09}
presented an $O(n \log n)$-time algorithm
that computes a minimum-width axis-parallel square annulus
with the matching $\Omega(n \log n)$ lower bound.
Bae~\cite{b-cmwsaao-18} presented the first $O(n^3 \log n)$-time algorithm
for the square annulus problem in arbitrary orientation,
and later improved it to $O(n^3)$ time~\cite{b-marsap-19X}.
In this paper, we further improve it to $O(n^2 \log n)$ time.

The rest of the paper is organized as follows:
In Section~\ref{sec:pre}, we introduce some preliminaries and definitions.
We define the Voronoi diagram and collects basic and essential observations
related to empty squares and the diagram in Section~\ref{sec:sq_VD}.
We then bound the lower and upper bounds on the number of
empty squares with four points on the boundary or with four contact pairs,
so answer our combinatorial question in Section~\ref{sec:4sq}.
Our algorithm that finds all those squares are presented in Section~\ref{sec:alg}
and its algorithmic applications are presented in Section~\ref{sec:mes}.

\section{Preliminaries} \label{sec:pre}
We consider the standard coordinate system with
the (horizontal) $x$-axis and the (vertical) $y$-axis in the plane $\Plane$.
We mean by the \emph{orientation} of any line, half-line, or line segment $\ell$
a unique real number $\theta \in [0,\pi)$ such that
$\ell$ is parallel to a rotated copy of the $x$-axis by $\theta$.

Through out the paper, we discuss squares in $\Plane$ in arbitrary orientation.
Each vertex and edge of a square will be called a \emph{corner} and a \emph{side}.
For any square $S$ in $\Plane$,
the \emph{orientation} of $S$ is a real number $\theta \in [0,\pi/2)$
such that the orientation of each side of $S$
is either $\theta$ or $\theta + \pi/2$.
We regard the set $\OS: = [0,  \pi/2)$ of all orientations of squares
as a topological space homeomorphic to a circle,
so $x = \theta$ modulo $\pi/2$ for any $x\in \Real$ and $\theta \in \OS$.

Each side of a square, as a subset of $\Plane$,
is assumed to include its incident corners.
We identify the four sides of a square $S$ by
the \emph{top}, \emph{bottom}, \emph{left}, and \emph{right} sides,
denoted by $\mt(S)$, $\mb(S)$, $\ml(S)$, and $\mr(S)$, respectively.
This identification is clear, regardless of the orientation of $S$
since it is chosen from $\OS = [0,\pi/2)$.
Similarly, we can say that a point lies to the \emph{left} of,
to the \emph{right} of,
\emph{above}, or \emph{below} another point in a fixed orientation.
The \emph{center} of a square is the intersection point of its two diagonals,
and its \emph{radius} is half its side length.


Let $P$ be a set of $n$ points in $\Plane$.
A square is called \emph{empty} if no point in $P$ lies in its interior.
An empty square may contain some points in $P$ on its boundary.
A pair $(p, \ma)$ of a point $p\in P$ and a side identifier
$\ma \in \{\mt, \mr, \mb, \ml\}$
is called a \emph{contact pair} of $S$ if $p \in \ma(S)$.
A set of contact pairs is called a \emph{contact type}.
If $\ct$ is the set of all contact pairs of $S$,
then we say that $\ct$ is the \emph{contact type} of $S$.
A \emph{contact point} $p\in P$ of $S$ is a point on a side of $S$, that is,
one belonging to a contact pair of $S$.
Each contact point $p \in P$ of $S$ may lie either on the relative interior of
a side of $S$ or at a corner of $S$.
In the former case, the contact point $p$ contributes to one contact pair of $S$,
while in the latter case, it to two contact pairs.


In this paper, we are interested in the number of empty squares
with four contact points under a proper general position.
We also ask an analogous question for empty squares
with four contact pairs.
We discuss our general position assumption
to answer these combinatorial questions
and devise our algorithmic results.

Before introducing our general position assumption on $P$,
some degenerate cases we would like to avoid are listed.
\begin{lemma} \label{lem:GP1}
 If there are infinitely many empty squares with four contact pairs,
 then there are $a,b,c,d\in P$ in convex position, 
 at most two of which may be identical, 
 such that one of the following holds:
 \begin{enumerate}[(1)] \denseitems
  \item Three distinct points in $\{a,b,c,d\}$ are collinear.
  \item There exists an empty square whose contact points are $a,b,c,d$
  and two edges of the convex hull of $\{a,b,c,d\}$
  are orthogonal or parallel.
  \item The convex hull of $\{a,b,c,d\}$
   forms a convex quadrilateral $abcd$
   such that its two diagonals $ac$ and $bd$ are orthogonal 
   and are of the same length.
 \end{enumerate}
 If there are infinitely many empty squares
 with four contact points, then
 there are four distinct points $a,b,c,d\in P$
 such that one of the above three conditions hold.
\end{lemma}
\begin{proof}
Assume that there are infinitely many empty squares with four or 
more contact pairs among $P$.
Let $K:=P\times\{\mt,\mr,\mb,\ml\}$ be the set of all possible contact pairs.
Since $K$ is finite, there are only finitely many possible contact types.
So, there must exist a contact type $\ct\subset K$ with four contact pairs
such that
there are infinitely many empty squares with the same contact type $\ct$.

Let $\Sq$ be the set of those empty squares with the same contact type $\ct$.
Let $(a,\ma_a)$, $(b, \ma_b)$, $(c, \ma_c)$,  and $(d,\ma_d)$ be
the four distinct contact pairs contained in $\ct$
for some $a,b,c,d\in P$ and $\ma_a, \ma_b, \ma_c, \ma_d\in\{\mt,\mr,\mb,\ml\}$.
Note that the cardinality of the set $\{a,b,c,d\}$ is either three or four,
since $\Sq$ is an infinite set.
Thus, two of $a,b,c,d\in P$ may be identical.
We consider three cases as follows: for any $S\in \Sq$,
every side of $S$ contains exactly one point in $\{a, b, c, d\}$,
one side of $S$ contains exactly two distinct points in $\{a, b, c, d\}$, or
one side of $S$ contains three or more distinct points in $\{a,b,c,d\}$.
In either case, the points $\{a,b,c,d\}$ are in convex position,
since they lie on the boundary of a square.
If this is the last case, then this is exactly case (1) where
three of $a,b,c,d,$ are collinear, so we are done.

Suppose the first case where every side of $S\in\Sq$ contains one point.
Assume that points $a$, $b$, $c$, and $d$ lie on
the top, right, bottom and left sides of $S$, in this order, respectively.
Observe, on one hand, that any other square $S' \in \Sq$ with $S'\neq S$
has a different orientation from that of $S$,
since any combination of scaling and translation $S$ without rotation
will lose a contact pair from $\ct$.
On the other hand, consider the rectangle $R(\theta)$ in orientation $\theta$
such that $a$ lies on its top side, $b$ on its right side, $c$ on its bottom side,
and $d$ on its left side.
Note that if the orientation of $S$ is $\theta_0$, then we have $S = R(\theta_0)$
and $R(\theta)$ is well defined for any $\theta$ sufficiently close to $\theta_0$.
Since those infinitely many squares in $\Sq$ have the same contact type $\ct$,
we conclude that for any $S' \in \Sq$,
$S' = R(\theta)$ for some $\theta$ sufficiently close to $\theta_0$.
Now, observe that the height and the width of $R(\theta)$ are represented as
sinusoidal functions of $\theta$:
\[ |ac| \cdot |\sin(\theta + \alpha)| \quad \text{and} \quad 
 |bd| \cdot |\sin(\theta + \beta + \pi/2)|,\]
where $\alpha$ and $\beta$ denote the orientations of $ac$ and $bd$, respectively.
For any $\theta$ sufficiently close to $\theta_0$, we have the equality
\[ |ac|\cdot|\sin(\theta + \alpha)| = |bd|\cdot|\sin(\theta + \beta + \pi/2)|.\]
This implies that $|ac| = |bd|$ and $\beta - \alpha = \pi/2$,
hence this is case (3) of the lemma.

Lastly, suppose the second case where one side of $S \in \Sq$ contains
two distinct points in $\{a,b,c,d\}$.
Assume that $c$ and $d$ are distinct and lie on the bottom side of $S$.
Points $a$ and $b$ are then distributed on the other three sides of $S$.
In this case, observe that all squares in $\Sq$ have the same orientation,
as the two points $c$ and $d$, by the corresponding contact pairs,
fix the orientation.
If $a$ and $b$ lie on different sides of $S$,
then we do not have a spare degree of freedom,
so there exists a unique square $S$ with the contact type $\ct$.
This is not the case by our assumption that $\Sq$ is an infinite set.
Hence, both $a$ and $b$ lie on a common side of $S$
other than the bottom side on which $c$ and $d$ lie.
If $a$ and $b$ lie on the left or the right side, then
two segments $ab$ and $cd$ are orthogonal;
otherwise, if $a$ and $b$ lie on the top side, 
then two segments $ab$ and $cd$ are parallel.
This means we have case (2).

If there are infinitely many empty squares with four contact points,
then there exists a contact type $\ct$ with $|\ct|=4$
such that there are infinitely many empty squares with four contact points
and the same contact type $\ct$, as done above.
In this case, the four contact points $\{a, b, c, d\}$ in $\ct$
are distinct.
This proves the second part of the lemma.
\end{proof}

Thus, we want to avoid such a configuration in $P$ of the above three cases.
In addition, consider an empty square $S$ with five contact pairs.
Then, there always exists a local continuous transformation
that makes $S$ lose one contact pair
and results in infinitely many empty squares with four contact pairs.
The above discussion suggests us the following \emph{general position assumption}:
\begin{quote} \textit{%
 There is no square in arbitrary orientation with
 five or more contact pairs among $P$.
 }
\end{quote}
Any square in arbitrary orientation has four degrees of freedom,
namely, two coordinates of its center, its radius, and its orientation.
This implies that four non-redundant contact pairs are enough
to determine a square.
Hence, the above assumption is indeed about a general position,
and can be achieved by applying any standard perturbation technique for $P$.
Observe that the three cases listed in Lemma~\ref{lem:GP1} are avoided
by our general position assumption on $P$.
\begin{lemma} \label{lem:GP2}
 Suppose that $P$ is in general position in our sense.
 Then, there are finitely many empty squares with four contact pairs
 among $P$.
\end{lemma}
\begin{proof}
Assume that $P$ is in general position,
and suppose that there are infinitely many empty squares with four contact pairs
among $P$.
Then, Lemma~\ref{lem:GP1} guarantees the existence of points
$a,b,c,d\in P$, two of which may be identical,
such that one of the three cases holds.
We handle each case.
\begin{enumerate}[(1)] \denseitems
\item If three distinct points of $a,b,c,d$ are collinear,
then a square with the side determined by
the minimal line segment containing the three points has at least five contact pairs.
So, we get a contradiction.
\item Suppose that there exists an empty square $S$ with contact points
$a,b,c,d$ and two edges of the convex hull of $\{a,b,c,d\}$
are orthogonal or parallel.
This means that two segments, say $ab$ and $cd$, 
lie on two distinct sides of the square $S$.
If $ab$ and $cd$ are parallel, then we observe that
$a, b, c,d$ should be all distinct.
We slide $S$ keeping $ab$ and $cd$
on its sides until one of the points $\{a, b, c, d\}$ is hit by a corner of $S$.
The resulting square now has five contact pairs, a contradiction.
Otherwise, if $ab$ and $cd$ are orthogonal, then
we assume that $ab$ is contained in the left side $\ml(S)$ of $S$
and $cd$ is contained in the bottom side $\mb(S)$ of $S$.
In this case, we shrink $S$ towards the bottom-left corner
until the top-left corner or the bottom-right corner hits
one of the points $\{a, b, c, d\}$.
The resulting square now has five contact pairs, a contradiction.
\item In the last case, the convex hull of $\{a, b, c, d\}$
 forms a convex quadrilateral $abcd$ such that
 its two segments $ac$ and $bd$ have the same length
 and make the right angle.
 Let $\theta_0\in [0,\pi)$ and $\theta'_0\in [0,\pi)$ 
 be the orientations of $ac$ and $bd$, respectively.
 Without loss of generality, assume that
 $\theta_0 < \theta'_0$,
 so $\theta'_0 - \theta_0 = \pi/2$.
 Consider the square $S_0$ whose orientation is $\theta'_0$ and
 each side of whose contains exactly one point in $\{a,b,c,d\}$.
 Note that $S_0$ is uniquely defined.
 As discussed in the proof of Lemma~\ref{lem:GP1},
 the rectangle defined by the four points on each side
 in any orientation sufficiently close to $\theta_0$ forms a square.
 We rotate $S_0$ by continuously increasing or decreasing $\theta$ from
 $\theta = \theta_0$, until one of the four points is located at a corner.
 Then, the resulting square has five contact pairs,
 a contradiction.
\end{enumerate}
In either case, we get a contradiction,
so the lemma is proved.
\end{proof}


From now on, we assume that $P$ is in general position as discussed above.
Consider any empty square $S$ of contact type $\ct$.
We call $S$ an \emph{$m$-square} or \emph{$(m,k)$-square}
if $m = |\ct|$ and $k$ is the number of contact points in $\ct$.
A side of $S$ is called \emph{pinned}
if it contains a point in $P$, so it is involved in some contact pair in $\ct$;
or \emph{stapled} if it contains two distinct points in $P$.
If $S$ has a stapled side, then $S$ is called \emph{stapled}.
From the general position assumption,
along with Lemmas~\ref{lem:GP1} and~\ref{lem:GP2},
there are no three or more contact pairs in $\ct$
involving a common side of $S$,
and there is at most one stapled side of $S$.

Since we are interested in $4$-squares,
we classify all possible types of $4$-squares
under the symmetry group of the square.
There are $10$ of them, classified by
the number of contact points and then the distribution of contact points
onto the square sides.
See \figurename~\ref{fig:4sq_type}.

\begin{figure}[tbh]
\begin{center}
\includegraphics[width=.95\textwidth]{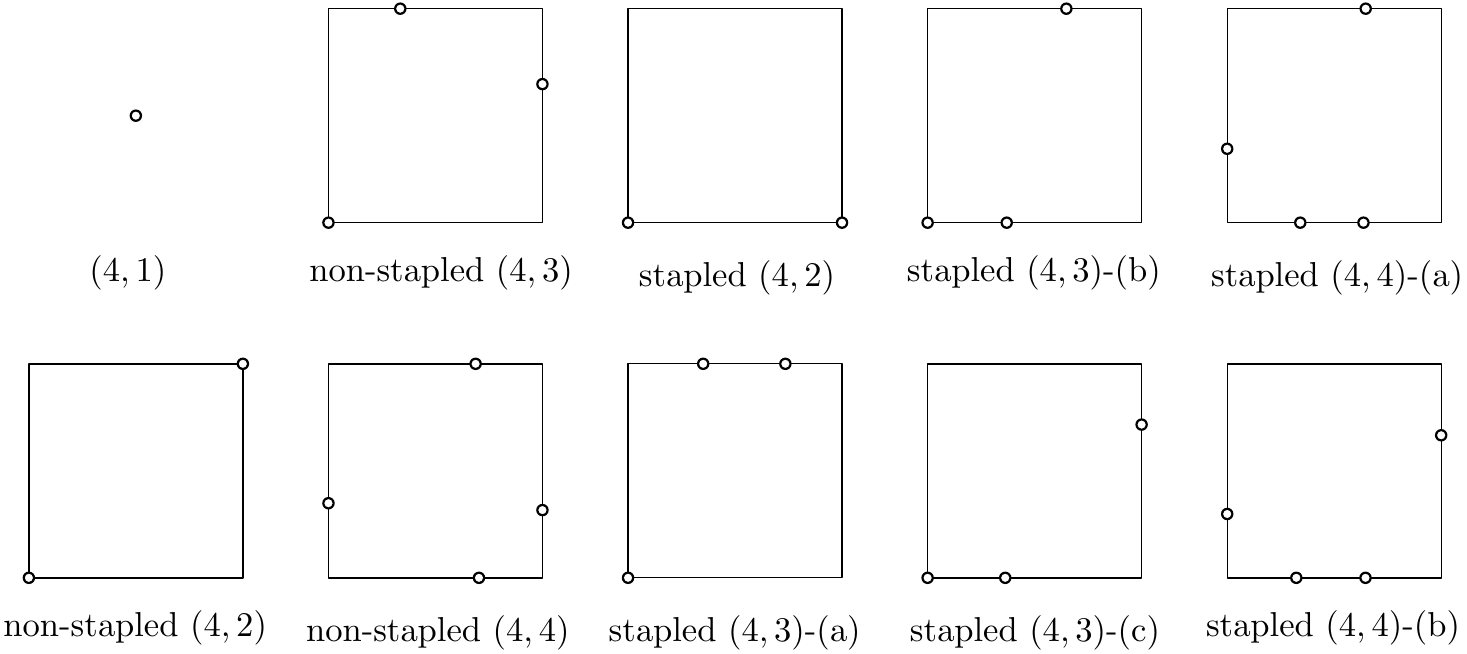}
\end{center}
\caption{Illustration of $10$ types of $4$-squares and the type names.
 Small circles on the boundary of each square depict its contact points in $P$.
 The four left ones are non-stapled types and
 the six right ones are stapled types.
 }
\label{fig:4sq_type}
\end{figure}

Suppose that $S$ is a $(4,k)$-square with contact type $\ct$.
Then, the type of $S$ is determined as follows:
\begin{itemize} \denseitems
 \item \textit{$(4,1)$-type:}
 In case of $k=1$, $S$ is degenerate to a point $p\in P$
 with $\ct = \{p\} \times \{\mt,\mr,\mb,\ml\}$.
 We call any $(4,1)$-square a \emph{trivial} square since its radius is zero
 and its orientation is undefined.

 \item \textit{$(4,2)$-types:}
 In case of $k=2$, two points $p, q\in P$ should lie on two corners of $S$.
 If $S$ is stapled, then $S$ is of \emph{stapled $(4,2)$-type};
 otherwise, $S$ is of \emph{non-stapled $(4,2)$-type}.

 \item \textit{$(4,3)$-types:}
 If $k=3$, then a contact point $p\in P$ should lie on a corner of $S$
 and the other two $q, r \in P$ lie in the interior of some sides of $S$.
 Without loss of generality, assume that $p$ lies on
 the bottom-left corner of $S$, so $(p,\mb), (p,\ml) \in \ct$.

 If $q$ and $r$ lie on distinct sides other than the bottom and left sides,
 then $S$ is of \emph{non-stapled $(4,3)$-type}.
 Otherwise, $S$ is stapled and
 there are three sub-types for stapled $(4,3)$-type.
 If both $q$ and $r$ lie on a common side of $S$,
 then the stapled side of $S$ is neither the bottom side nor the left side
 by the general position assumption;
 in this case, $S$ is said to be of \emph{stapled $(4,3)$-(a)-type}.
 Otherwise, one of $q$ and $r$
 lies on the bottom side or on the left side.
 Without loss of generality, assume that $(r,\mb)\in \ct$.
 Then, there are two cases: either $(q,\mt) \in \ct$ or $(q,\mr)\in\ct$.
 We say that $S$ is of \emph{stapled $(4,3)$-(b)-type} if $(q,\mt)\in\ct$;
 or $S$ is of \emph{stapled $(4,3)$-(c)-type} if $(q,\mr)\in\ct$.

 \item \textit{$(4,4)$-types:}
 If every side of $S$ is pinned, then $S$ is of \emph{non-stapled $(4,4)$-type}.
 Otherwise, if $S$ is stapled, then two contact points $p, q\in P$
 lie on a common side of $S$, say the bottom side,
 so $(p,\mb), (q, \mb)\in\ct$.
 There are two sub-types of stapled $(4,4)$-type:
 $S$ is of \emph{stapled $(4,4)$-(a)-type} if
 the top side of $S$ is pinned by one of the other two contact points;
 otherwise, $S$ is of \emph{stapled $(4,4)$-(b)-type}.
\end{itemize}

Throughout the paper,
we are not interested in trivial $(4,1)$-squares in most cases.
Hereafter, we thus mean by a $4$-square
a nontrivial $4$-square, that is, a $(4,k)$-square with $k \geq 2$,
unless stated otherwise.

\section{Empty Squares and the Voronoi Diagram} \label{sec:sq_VD}

The empty squares among $P$ are closely related to the Voronoi diagram of $P$
under the $L_\infty$ distance.
It is well known that each vertex of the $L_\infty$ diagram corresponds to
an axis-parallel $3$-square and each edge to the locus
of axis-parallel $2$-squares with a common contact type,
unless there is any axis-parallel $4$-square.
Indeed, the diagram explains a complete geometric relation
among all the empty squares in orientation $0\in\OS$.
In this section, we extend this knowledge for
empty squares in arbitrary orientation.
Consequently, we collect several essential properties of empty squares
in terms of the Voronoi diagram,
based on which we will be able to bound the number of $4$-squares
and to present an efficient algorithm that computes all $4$-squares.

\subsection{Definition of Voronoi diagrams}
Let $P$ be a given set of $n$ points in general position
as discussed in Section~\ref{sec:pre}.
For each $\theta \in \OS$,
we define $\VD(\theta)$ to be the $L_\infty$ Voronoi diagram of $P$
with the axes rotated by $\theta$, or
equivalently, the Voronoi diagram of $P$
under the symmetric convex distance function $d_\theta$ based on
a unit square whose orientation is $\theta$.
The \emph{Voronoi region} of $p\in P$ in orientation $\theta$ is
 \[ \VR_p(\theta) := \{ x\in \Plane \mid d_\theta(x, p) < d_\theta(x, q), q\in P \setminus \{p\}\}.\]
For fixed $\theta\in\OS$, the diagram $\VD(\theta)$ is just
a rotated copy of the $L_\infty$ Voronoi diagram of a rotated copy of $P$.

The diagram $\VD(\theta)$ can also be defined in terms of empty squares.
More precisely, we view $\VD(\theta)$ as a plane graph
whose vertices $\hat{V}(\theta)$ and
edges $\hat{E}(\theta)$ are determined as follows:
\begin{itemize} \denseitems
\item The vertex set $\hat{V}(\theta)$ consists of the centers of all
empty squares in orientation $\theta$ with three or four pinned sides,
and a point \emph{at infinity}, denoted by $\hat{\infty}$.
\item An edge is contained in $\hat{E}(\theta)$ if and only if
it is a maximal set of centers of all empty squares in orientation $\theta$
having a common contact type with two pinned sides.
Each edge in $\hat{E}(\theta)$ is either a half-line or a line segment,
called \emph{unbounded} or \emph{bounded}, respectively.
In either case, any endpoint of an edge in $\hat{E}(\theta)$
is a vertex in $\hat{V}(\theta)$,
and each unbounded edge is incident to $\hat{\infty}$.
\end{itemize}
The vertices and edges of $\VD(\theta)$ are well defined
by empty squares with a certain number of pinned sides:
if three or four sides are pinned, then there is such a unique square in a fixed orientation;
if two sides are pinned, then we can slide or grow a corresponding square
whose center traces out an edge.

\begin{figure}[tbh]
\begin{center}
\includegraphics[width=.99\textwidth]{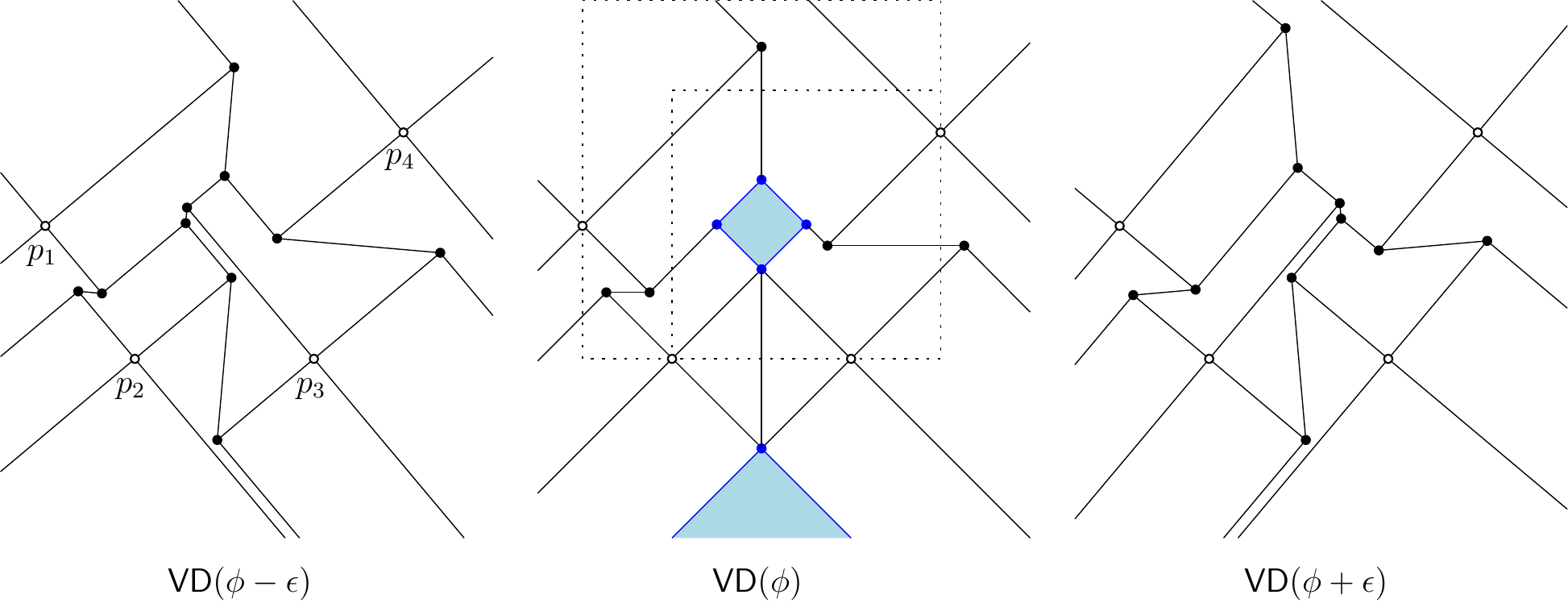}
\end{center}
\caption{Illustration of $\VD(\phi-\epsilon)$, $\VD(\phi)$,
and $\VD(\phi+\epsilon)$ for $P = \{p_1, \ldots, p_4\}$ and some $\phi\in\OS$.
 }
\label{fig:vd}
\end{figure}

See \figurename~\ref{fig:vd}.
Our definition of $\VD(\theta)$ is not very different from any standard one,
except it adds four more edges incident to each $p\in P$
such that each of them corresponds to $2$-squares with one corner
anchored at $p$ (hence, with two contact pairs and two pinned sides).
In this way, each point $p\in P$ is also a vertex in $\hat{V}(\theta)$
since $p$ is the center of a trivial $(4,1)$-square
with its four sides pinned.
The diagram $\VD(\theta)$ as a plane graph divides the plane
into its \emph{faces}.
Each face of $\VD(\theta)$ is the locus of centers
of empty squares having a common contact type with \emph{one pinned side};
hence, for every contact pair $(p,\ma)$,
there exists a unique face of $\VD(\theta)$ consisting of the centers
of $1$-squares with contact type $\{(p,\ma)\}$.
Therefore, the Voronoi region $\VR_p(\theta)$ of each $p\in P$
includes exactly four faces of $\VD(\theta)$.
On the other hand, there may exist some \emph{neutral} faces of $\VD(\theta)$
that do not belong to any Voronoi region $\VR_p(\theta)$,
if it corresponds to $2$-squares with a stapled side
(see the two shaded faces of $\VD(\phi)$ in color lightblue in \figurename~\ref{fig:vd}).
This is obviously a degenerate case
which has been avoided from most discussions about Voronoi diagrams
in the literature.
In this way, our definition of $\VD(\theta)$ completely represents
all cases of point set $P$,
even though there are four equidistant points in $P$ under $d_\theta$
or there are two points in $P$ such that $pq$ is in orientation $\theta$
or $\theta+\pi/2$.

The \emph{combinatorial structure} of $\VD(\theta)$ is represented by
its underlying graph $\VG(\theta) = (V(\theta), E(\theta))$,
called the \emph{Voronoi graph};
conversely, $\VD(\theta)$ is a plane embedding of $\VG(\theta)$.
More precisely, the vertices and edges of $\VG(\theta)$ are
described and identified as follows:
\begin{itemize} \denseitems
 \item Each vertex $v \in V(\theta)$ corresponds to $\hat{v} \in \hat{V}(\theta)$.
 In particular, the vertex at infinity, denoted by $\infty \in V(\theta)$,
 corresponds to $\hat{\infty} \in \hat{V}(\theta)$.
 Each $v\in V(\theta) \setminus \{\infty\}$
 is identified by the contact type $\ct_v$ of the square defining
 $\hat{v} \in \hat{V}(\theta)$.
 For completeness, we define $\ct_\infty := \emptyset$ 
 \item Each edge $e\in E(\theta)$ corresponds to $\hat{e} \in \hat{E}(\theta)$.
 Each edge $e = uv \in E(\theta)$ for $u, v \in V(\theta)$
 is identified by a triple $(\ct_u, \ct_v; \ct_e)$,
 where $\ct_e$ is the contact type of the squares defining
 $\hat{e}\in \hat{E}(\theta)$.
 If $e$ is bounded, then we have $\ct_e = \ct_u \cap \ct_v$.
\end{itemize}
Hence, for $\theta, \theta' \in \OS$,
two vertices $v \in V(\theta)$ and $v' \in V(\theta')$ are the same
if $\ct_v = \ct_{v'}$;
two edges $uv \in E(\theta)$ and $u'v' \in E(\theta')$ are the same
if $(\ct_u, \ct_v; \ct_{uv}) = (\ct_{u'}, \ct_{v'}; \ct_{u'v'})$.
We say that $\VD(\theta)$ and $\VD(\theta')$ are
\emph{combinatorially equivalent} if $\VG(\theta) = \VG(\theta')$.

\subsection{Basic properties}

For each vertex $v\in V(\theta)$, we call $v$ and
its embedding $\hat{v} \in \hat{V}(\theta)$ \emph{regular} if $|\ct_v| = 3$,
that is, its corresponding empty square is a $3$-square.
For each edge $e \in E(\theta)$,
we call $e$ and its embedding $\hat{e} \in \hat{E}(\theta)$
\emph{regular} if $|\ct_{e}| = 2$.

Each edge $e \in E(\theta)$ is called \emph{sliding}
if the two pinned sides in $\ct_{e}$ are parallel,
or \emph{growing}, otherwise.
From the properties of the $L_\infty$ Voronoi diagram,
we observe that any sliding edge is in orientation $\theta$ or $\theta + \pi/2$,
while any growing edge is in orientation $\theta + \pi/4$ or $\theta + 3\pi/4$
(modulo $\pi$).
For $e \in E(\theta)$,
we regard each growing edge to be directed in which
its corresponding square is growing.


\begin{figure}[tbh]
\begin{center}
\includegraphics[width=.99\textwidth]{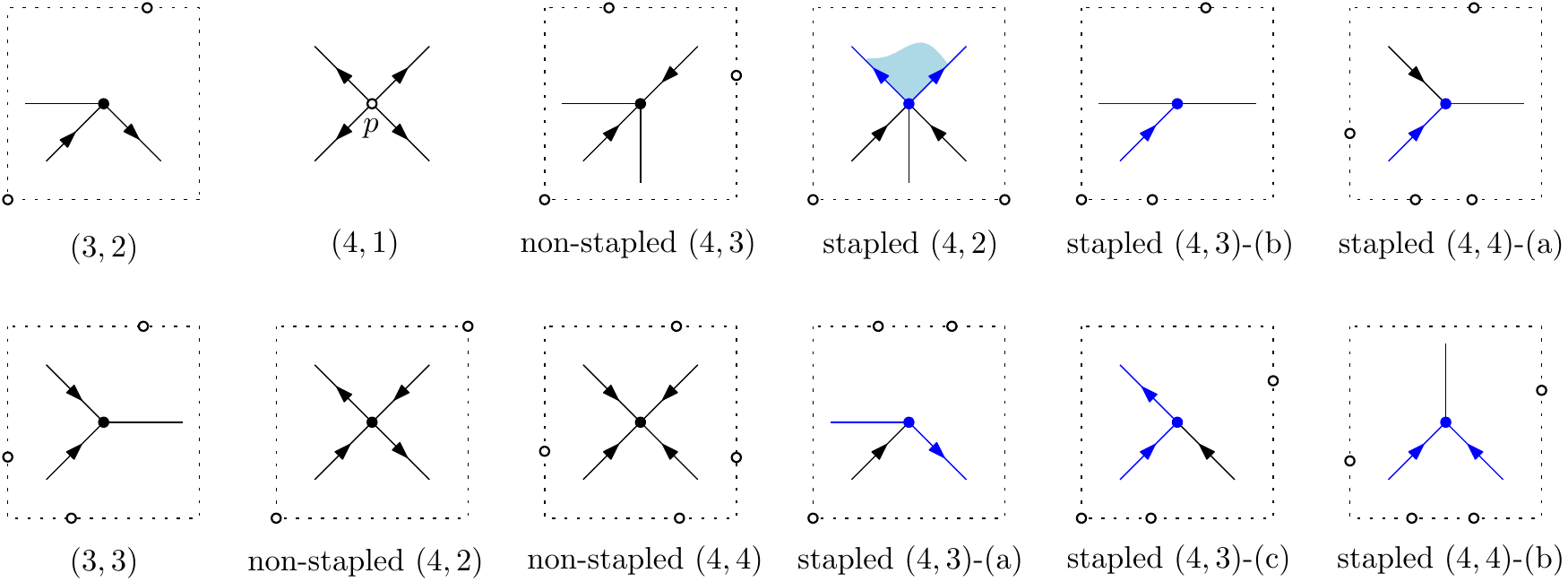}
\end{center}
\caption{Illustration of the $12$ vertex types of $\VD(\theta)$
 labeled with their type names.
 Dotted squares and small circles on them depict
 an empty square corresponding to each vertex type and its contact points.
 Dots and line segments depict vertices and all incident edges to each vertex;
 in black if regular, or in blue, otherwise.
 The arrows on edges depict the direction in which the corresponding square
 grows.
 The shaded area in the stapled $(4,2)$-type depicts a neutral face
 which does not belong to any Voronoi region $\VR_p(\theta)$ for $p\in P$.
 (We keep the above convention in every figure throughout the paper.)
 Note the right six types are stapled, while the left six are not.
 }
\label{fig:vertextype}
\end{figure}

For each vertex $v \in V(\theta)$ with $v\neq \infty$,
the local structure of the diagram $\VD(\theta)$ around $\hat{v}$
is completely determined by its contact type $\ct_v$.
From the possible contact types of $3$- and $4$-squares,
we classify all vertices in $V(\theta) \setminus \{\infty\}$
into $12$ vertex types.
\begin{lemma} \label{lem:vertextype}
 There are $12$ types for the vertices of $\VD(\theta)$
 as illustrated in \figurename~\ref{fig:vertextype}.
 For any vertex $v\in V(\theta)\setminus\{\infty\}$,
 the contact types $\ct_e$ of all edges $e \in E(\theta)$ incident to $v$
 can be obtained only from the contact type $\ct_v$
 without knowing the other incident vertex of $e$.
\end{lemma}
\begin{proof}
Consider any vertex $v\in V(\theta)$ with $v\neq \infty$
and its contact type $\ct_v$.
Let $S$ be the square with contact type $\ct_v$ in orientation $\theta$
whose center is $\hat{v} \in \hat{V}(\theta)$.
If $v$ is regular, then its contact type $\ct_v$
consists of three contact pairs and two or three contact points.
Hence, $S$ is either a $(3,2)$-square or a $(3,3)$-square.
Respectively, $v$ is of $(3,2)$-type or $(3,3)$-type.

If $v$ is non-regular, then $S$ is a $4$-square.
According to the type of $S$ as classified in \figurename~\ref{fig:4sq_type},
the type of $v$ is distinguished.
Thus, $v$ is of one of the $10$ types if it is non-regular.

For each of these $12$ cases,
we can enumerate the contact types $\ct_e$ of all incident edges $e$
by locally transforming the corresponding square $S$.
In general, for each vertex $v$ and its contact type $\ct_v$,
we enumerate all possible subsets $\ct$ of $\ct_v$, whose number is constant,
and check if there exists a locally continuous transformation
of squares with contact type $\ct$ from $S$.
Below, we show a concrete analysis for every vertex type.
This analysis is indeed lengthy and tedious, while it is worth
giving all the details for complete exposure.
\begin{itemize}
 \item Suppose that $v$ is of $(3,2)$-type, and assume that
 $\ct_v = \{(p, \mb), (p,\ml), (q,\mt)\}$ for some $p,q\in P$.
 So, $p$ lies at the bottom-left corner and $q$ lies in the relative interior
 of the top side of $S$.
 See the first case in \figurename~\ref{fig:vertextype}.
 In this case, we have three incident edges as follows.
 First, if we slide $S$ leftwards, then we lose a contact pair,
 namely, $(p, \ml)$, with its center tracing out a sliding edge 
 in $\hat{E}(\theta)$ by definition.
 Second, if we shrink $S$ keeping two contact pairs
 $(p,\mb)$ and $(p,\ml)$, then its center traces out a growing edge
 towards $p$, whose growing direction is towards $v$.
 Third, if we grow $S$ keeping two contact pairs
 $(p,\ml)$ and $(q, \mt)$, then its center traces out a growing edge
 towards the bottom-right corner of $S$.
 There are no more incident edge, and the three incident edges
 are all regular since their contact type consists of two contact pairs.
 \item If $v$ is of $(3,3)$-type, then
 we assume that $\ct_v = \{(p,\ml),(q,\mb),(r,\mt)\}$ for $p,q,r\in P$.
 See the $(3,3)$-type in \figurename~\ref{fig:vertextype}.
 As done above, we check every local transformation of $S$
 by sliding, growing, or shrinking it.
 We then observe that there are three edges $e_1,e_2,e_3$
 incident to $v$ with contact types:
 $\ct_{e_1} = \{(q,\mb),(r,\mt)\}$,
 $\ct_{e_2} = \{(p,\ml),(r,\mt)\}$, and
 $\ct_{e_3} = \{(p,\ml),(q,\mb)\}$.
 Note that $e_1$ is sliding, while $e_2$ and $e_3$ are growing
 edges whose growing direction is towards $v$.
 All the three edges are regular.
 \item If $v$ is of $(4,1)$-type, then $v$ is a point $p\in P$.
 See the $(4,1)$-type in \figurename~\ref{fig:vertextype}.
 In this case, there are four edges corresponding to
 the $(2,1)$-squares having $p$ at one corner.
 All the four edges are growing, directed outwards from $p$.
 \item If $v$ is of non-stapled $(4,2)$-type, then
 we assume that $\ct_v = \{(p,\ml),(p,\mb),(q,\mt),(q,\mr)\}$ for $p,q\in P$.
 See the non-stapled $(4,2)$-type in \figurename~\ref{fig:vertextype}.
 As done above, we check every local transformation of $S$
 by sliding, growing, or shrinking it.
 We then observe that there are four edges $e_1,e_2,e_3,e_4$
 incident to $v$ with contact types:
 $\ct_{e_1} = \{(p,\mb),(q,\mr)\}$,
 $\ct_{e_2} = \{(p,\ml),(q,\mt)\}$, 
 $\ct_{e_3} = \{(p,\mb),(p,\ml)\}$, and
 $\ct_{e_4} = \{(q,\mt),(q,\mr)\}$.
 Note that all of the four edges are growing such that
 $e_1$ and $e_2$ are directed outwards from $v$,
 while $e_3$ and $e_4$ are directed towards $v$.
 Since $S$ is empty, $e_3$ is also incident to $p$ and $e_4$ is to $q$.
 All the three edges are regular.
 \item If $v$ is of non-stapled $(4,3)$-type, then
 we assume that $\ct_v = \{(p,\ml),(p,\mb),(q,\mt),(r,\mr)\}$ for $p,q,r\in P$.
 See the non-stapled $(4,3)$-type in \figurename~\ref{fig:vertextype}.
 In this case,
 we observe that there are four edges $e_1,e_2,e_3,e_4$
 incident to $v$ with contact types:
 $\ct_{e_1} = \{(p,\mb),(q,\mt)\}$,
 $\ct_{e_2} = \{(p,\ml),(q,\mr)\}$, 
 $\ct_{e_3} = \{(p,\mb),(p,\ml)\}$, and
 $\ct_{e_4} = \{(q,\mt),(r,\mr)\}$.
 Note that $e_1$ and $e_2$ are sliding,
 while $e_3$ and $e_4$ are growing, directed towards $v$.
 Since $S$ is empty, $e_3$ is also incident to $p$.
 All the three edges are regular.
 \item If $v$ is of non-stapled $(4,4)$-type, then
 we assume that $\ct_v = \{(p_1,\ml),(p_2,\mb),(p_3,\mt),(p_4,\mr)\}$ 
 for distinct $p_1,p_2,p_3,p_4\in P$.
 See the non-stapled $(4,4)$-type in \figurename~\ref{fig:vertextype}.
 In this case,
 we observe that there are four edges $e_1,e_2,e_3,e_4$
 incident to $v$ with contact types:
 $\ct_{e_1} = \{(p_1,\ml),(p_2,\mb)\}$,
 $\ct_{e_2} = \{(p_2,\mb),(p_3,\mr)\}$, 
 $\ct_{e_3} = \{(p_3,\mr),(p_4,\mt)\}$, and
 $\ct_{e_4} = \{(p_4,\mt),(p_1,\ml)\}$.
 Note that each of the four edges is growing, directed towards $v$.
 All the three edges are regular.
 \item If $v$ is of stapled $(4,2)$-type, then we assume that
 $\ct_v = \{(p,\mb), (p,\ml), (q,\mr), (q,\mb)\}$ for $p,q\in P$,
 as for the stapled $(4,2)$-type shown in \figurename~\ref{fig:vertextype}.
 Checking all possibilities,
 we have five edges incident to $v$ with contact types:
 $\ct_{e_1}=\{(p,\ml), (q,\mr)\}$,
 $\ct_{e_2}=\{(p,\mb), (p,\ml)\}$,
 $\ct_{e_3}=\{(q,\mr), (q,\mb)\}$,
 $\ct_{e_4}=\{(p,\mb), (p,\ml), (q,\mb)\}$, and
 $\ct_{e_5}=\{(p,\mb), (q,\mr), (q,\mb)\}$.
 Note that $e_1$ is sliding and the other four are growing.
 Edges $e_2$ and $e_3$ are directed towards $v$ while
 $e_4$ and $e_5$ are directed outwards from $v$.
 The first three edges are regular while the last two are non-regular.
 \item If $v$ is of stapled $(4,3)$-(a)-type, then we assume that
 $\ct_v = \{(p,\mb), (p,\ml), (q,\mt), (r,\mt)\}$ for $p,q,r\in P$,
 as for the stapled $(4,3)$-(a)-type shown in \figurename~\ref{fig:vertextype}.
 In this case,
 there are three edges $e_1,e_2, e_3$ incident to $v$ with contact types:
 $\ct_{e_1}=\{(p,\mb), (p,\ml)\}$,
 $\ct_{e_2}=\{(p,\ml), (q,\mt), (r,\mt)\}$, and
 $\ct_{e_3}=\{(p,\mb), (q,\mt), (r,\mt)\}$.
 Note that $e_1$ is a regular and growing edge, directed towards $v$,
 $e_2$ is a non-regular growing edge, directed outwards from $v$,
 and $e_3$ is a non-regular sliding edge.
 \item If $v$ is of stapled $(4,3)$-(b)-type, then we assume that
 $\ct_v = \{(p,\mb), (p,\ml), (q,\mb), (r,\mt)\}$ for $p,q,r\in P$,
 as for the stapled $(4,3)$-(b)-type shown in \figurename~\ref{fig:vertextype}.
 In this case,
 there are three edges $e_1,e_2, e_3$ incident to $v$ with contact types:
 $\ct_{e_1}=\{(q,\mb), (r,\mt)\}$,
 $\ct_{e_2}=\{(p,\mb), (q,\mb), (r,\mt)\}$, and
 $\ct_{e_3}=\{(p,\ml), (p,\mb), (q,\mb)\}$.
 Note that $e_1$ is a regular and sliding edge,
 $e_2$ is a non-regular sliding edge, and
 $e_3$ is a non-regular growing edge, directed towards $v$.
 \item If $v$ is of stapled $(4,3)$-(c)-type, then we assume that
 $\ct_v = \{(p,\mb), (p,\ml), (q,\mb), (r,\mr)\}$ for $p,q,r\in P$.
 See the stapled $(4,3)$-(b)-type of \figurename~\ref{fig:vertextype}.
 In this case,
 there are three edges $e_1,e_2, e_3$ incident to $v$ with contact types:
 $\ct_{e_1}=\{(q,\mb), (r,\mr)\}$,
 $\ct_{e_2}=\{(p,\mb), (q,\mb), (r,\mr)\}$, and
 $\ct_{e_3}=\{(p,\ml), (p,\mb), (q,\mb)\}$.
 Note that all the three edges are growing.
 Edge $e_1$ is regular, while the other two are non-regular.
 Edge $e_2$ is directed outwards from $v$, while
 the other two are directed towards $v$.
 \item If $v$ is of stapled $(4,4)$-(a)-type, then we assume that
 $\ct_v = \{(p,\mb), (q,\mb), (r_1,\ml), (r_2,\mt)\}$ for $p,q,r_1,r_2\in P$,
 See the stapled $(4,4)$-(a)-type of \figurename~\ref{fig:vertextype}.
 In this case,
 there are three edges $e_1,e_2, e_3$ incident to $v$ with contact types:
 $\ct_{e_1}=\{(r_1,\ml), (r_2,\mt)\}$,
 $\ct_{e_2}=\{(p,\mb), (q,\mb), (r_1,\ml)\}$, and
 $\ct_{e_3}=\{(p,\mb), (q,\mb), (r_2,\mt)\}$.
 Note that $e_1$ is a regular growing edge, directed towards $v$,
 $e_2$ is a non-regular growing edge, directed towards $v$, and
 $e_3$ is a non-regular sliding edge.
 \item If $v$ is of stapled $(4,4)$-(b)-type, then we assume that
 $\ct_v = \{(p,\mb), (q,\mb), (r_1,\ml), (r_2,\mr)\}$ for $p,q,r_1,r_2\in P$,
 See the stapled $(4,4)$-(b)-type of \figurename~\ref{fig:vertextype}.
 In this case,
 there are three edges $e_1,e_2, e_3$ incident to $v$ with contact types:
 $\ct_{e_1}=\{(r_1,\ml), (r_2,\mr)\}$,
 $\ct_{e_2}=\{(p,\mb), (q,\mb), (r_1,\ml)\}$, and
 $\ct_{e_3}=\{(p,\mb), (q,\mb), (r_2,\mr)\}$.
 Note that $e_1$ is a regular sliding edge, while
 $e_2$ and $e_3$ are non-regular growing edges, directed towards $v$.
\end{itemize}
Consequently, we show that a local configuration of each vertex
of $\VD(\theta)$ can be obtained in a local way.
The details about the edges incident to a vertex of each type,
we have obtained above,
are illustrated in \figurename~\ref{fig:vertextype}.
\end{proof}

Consider all squares with contact type $\ct_e$
defining $\hat{e} \in \hat{E}(\theta)$.
If $e$ is sliding, then all these squares have the same radius;
if $e$ is growing, then their radius grows along $\hat{e}$.
The following lemma is an immediate observation.
\begin{lemma} \label{lem:edge_vertex}
 Let $e = uv \in E(\theta)$ be any bounded edge for any $\theta \in \OS$,
 and $S_u$ and $S_v$ be the empty squares in $\theta$
 with contact types $\ct_u$ and $\ct_v$, respectively.
 Then, the union $R$ of all squares in $\theta$ with contact type $\ct_e$
 is equal to $S_u \cup S_v$.
 More specifically, if $e$ is sliding, then
 $R = S_u \cup S_v$ forms a rectangle;
 if $e$ is growing, then one of $S_u$ and $S_v$ completely
 contains the other, so $R$ is a square.
\end{lemma}
Boissonnat et al.\@ have proved a generalized version of the above
Lemma~\ref{lem:edge_vertex},
see Lemmas 8.1--8.5 in~\cite{bsty-vdhdcpdf-98},
so we omit its proof.

\begin{figure}[tbh]
\begin{center}
\includegraphics[width=.65\textwidth]{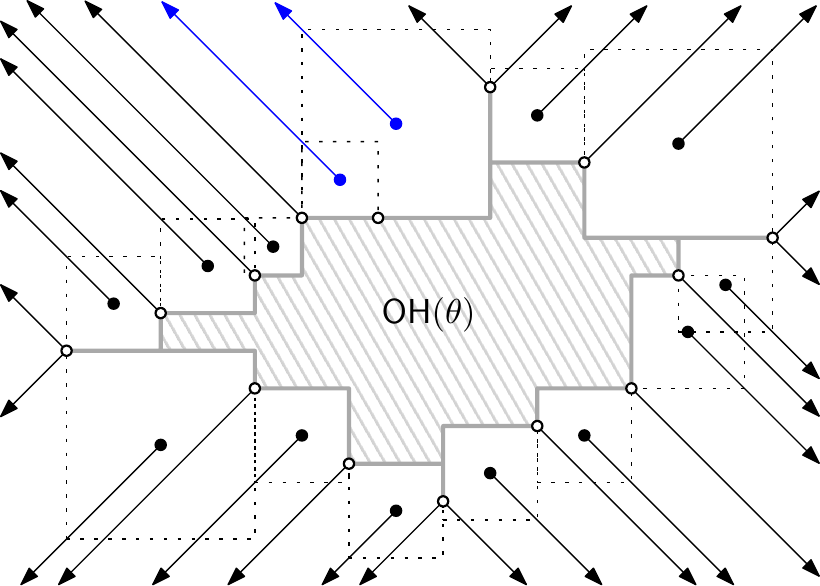}
\end{center}
\caption{Illustration to Lemma~\ref{lem:unbounded_och}.
 The orthogonal convex hull $\OCH(\theta)$ of $P$ (shaded region)
 and the four staircases (gray thick lines) describe
 all unbounded edges (arrows) and their incident vertices of $\VD(\theta)$
 (small circles and dots).
 }
\label{fig:och_unbounded}
\end{figure}

The above lemma indeed extends to the case of unbounded edges.
Consider any unbounded edge $e\in E(\theta)$ and the union $R$ of all squares
as declared in Lemma~\ref{lem:edge_vertex}.
Observe that $R$ forms an empty unbounded quadrant rotated by $\theta$
and
the empty quadrant $R$ has two or three contact points in $P$.
This tells us a relation between unbounded edges and
the orthogonal convex hull of $P$.
The \emph{orthogonal convex hull} of point set $P$ is defined to be
the minimal subset of $\Plane$ such that any vertical or horizontal line
intersects it in at most one connected component.
It is known that the orthogonal convex hull is obtained by
subtracting all empty quadrants (of four directions) from the whole plane $\Plane$
and its boundary is represented by four monotone chains,
called the \emph{staircases}.
For $\theta \in \OS$, let $\OCH(\theta)$ denote
the orthogonal convex hull of $P$ with the axes rotated by $\theta$.
See \figurename~\ref{fig:och_unbounded} and
Bae et al.~\cite{blacc-cmarchls-09} for more details on $\OCH(\theta)$,
including the precise definition of $\OCH(\theta)$ and the staircases.
\begin{lemma} \label{lem:unbounded_och}
 For any $\theta\in\OS$,
 all the unbounded edges and their incident vertices of $\VD(\theta)$
 are explicitly described by $\OCH(\theta)$, in the sense that
 if $v \in V(\theta)\setminus\{\infty\}$ is
 a vertex incident to an unbounded edge,
 then either
 \begin{enumerate}[(i)] \denseitems
  \item we have $\hat{v} = p\in P$ and $p$ coincides with 
  a vertex of $\OCH(\theta)$, or
  \item its contact points in $\ct_v$ appear consecutively
 in a staircase of $\OCH(\theta)$.
 \end{enumerate}
\end{lemma}
\begin{proof}
First, suppose that $v$ is incident to an unbounded edge $e\in E(\theta)$.
Since $e$ is growing, its orientation is either $\pi/4$ or $3\pi/4$.
Assume that the orientation of $e$ is $3\pi/4$ and
the corresponding empty quadrant $R$ for $e$ is unbounded upwards and leftwards.
This means that
$\ct_e$ contains two contact pairs $(p, \mb)$ and $(q, \mr)$
for some $p, q\in P$.
Note that $e$ is growing and directed outwards from $v$.
There are only six possible vertex types for $v$
having such an incident edge
by Lemma~\ref{lem:vertextype} and \figurename~\ref{fig:vertextype};
namely,
$(3,2)$-type, $(4,1)$-type, stapled and non-stapled $(4,2)$-types,
stapled $(4,3)$-(a)-type, and stapled $(4,3)$-(c)-type.
In either case, there cannot be a contact point
in the relative interior of the left side and the top side
of the corresponding square by the emptiness of $R$.
Hence, the set of contact points of $\ct_v$ is the same
as that of $R$.
See \figurename~\ref{fig:och_unbounded}.

If $v$ is of $(4,1)$-type,
then $\hat{v} = p\in P$, $\ct_e = \{(p,\mb),(p,\mr)\}$,
and no other point in $P$ lies on the boundary on $R$.
Hence, we can slide $R$ slightly rightwards or downwards,
keeping it empty.
This implies that $p$ coincides with a vertex of $\OCH(\theta)$.

Otherwise, the number of contact points in $\ct_v$ is
two or three.
Since $R$ is an empty quadrant bounded by some points in $P$,
the contact points of $R$, or equivalently,
that of $\ct_v$, appear consecutively
along a staircase of $\OCH(\theta)$~\cite{blacc-cmarchls-09}.
In particular, if $v = p\in P$, then
$p$ lies on the boundary of $\OCH(\theta)$.
%
%
%
\end{proof}


Another implication of Lemma~\ref{lem:vertextype}
is that the degree of every vertex, except $\infty\in V(\theta)$,
is at least three and at most five.
This implies the linear complexity of $\VD(\theta)$
for any $\theta\in\OS$.
\begin{lemma} \label{lem:VD_complexity}
 For any $\theta \in \OS$,
 the number of vertices, edges, and faces of $\VD(\theta)$ is $\Theta(n)$.
\end{lemma}
\begin{proof}
The number of faces in $\VD(\theta)$ is at least $4n$,
since every contact pair $(p,\ma)$ for $p\in P$ and $\ma\in\{\mt,\mr,\mb,\ml\}$
has a corresponding face.
On the other hand, there can be at most $n$ neutral faces in $\VD(\theta)$.
Any neutral face of $\VD(\theta)$, if exists,
consists of the centers of $2$-squares in orientation $\theta$
with a common contact type $\ct$ such that $\ct = \{(p,\ma),(q,\ma)\}$
for some $p, q\in P$ with $p\neq q$ and some $\ma\in\{\mt,\mr,\mb,\ml\}$.
Hence, there should exist a pair of points $p,q\in P$ such that
segment $pq$ is in orientation $\theta$ or $\theta+\pi/2$.
Note that such a pair of points $p,q$ can be involved in
at most two neutral faces as shown in the stapled $(4,2)$-type in
\figurename~\ref{fig:vertextype}.
By our general position assumption,
each point in $P$ can be involved in at most one such pair,
so
there can be at most $n/2$ such pairs.
Hence, the number of faces in $\VD(\theta)$ is $\Theta(n)$.

By Lemma~\ref{lem:vertextype},
the degree of every vertex in $V(\theta) \setminus \{\infty\}$
is between three and five.
This, together with the Euler's formula about planar graphs,
implies that the number of vertices and edges is also $\Theta(n)$.
\end{proof}

\subsection{Combinatorial changes of $\VD(\theta)$ and $4$-squares}

Let $\Sq_4$ be the set of all nontrivial $4$-squares among $P$.
By our general position assumption on $P$ and Lemma~\ref{lem:GP2},
we know that $\Sq_4$ is finite.
Let $s_4 := |\Sq_4|$ be the number of $4$-squares among $P$.
For any orientation $\theta \in \OS$,
we call $\theta$ \emph{regular}
if there is no nontrivial $4$-square in orientation $\theta$,
or \emph{degenerate}, otherwise.
Since there are only finitely many, exactly $s_4$, $4$-squares,
all orientations $\theta \in \OS$ but at most $s_4$ of them are regular.

\begin{figure}[tbh]
\begin{center}
\includegraphics[width=.9\textwidth]{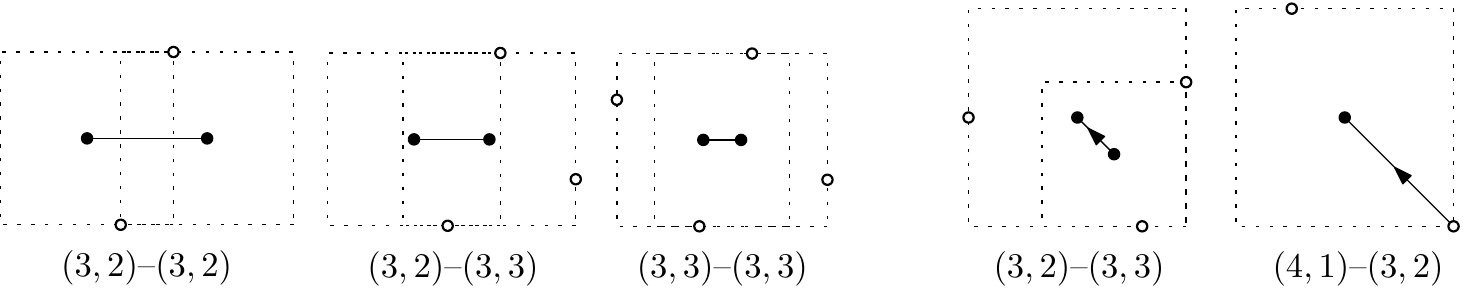}
\end{center}
\caption{Illustration of five types of bounded edges of $\VD(\theta)$
for a regular orientation $\theta\in \OS$,
depending on the types of incident vertices.
The first three are sliding and the last two are growing.
 }
\label{fig:edgetype}
\end{figure}

In a regular orientation $\theta$,
the diagram $\VD(\theta)$ has the following properties.
\begin{lemma} \label{lem:regularVD}
 For any regular orientation $\theta \in \OS$,
 every vertex in $V(\theta)$, except those in $P$ and $\infty$,
 is regular, and every edge in $E(\theta)$ is regular.
 There are five types of bounded edges, as shown in
 \figurename~\ref{fig:edgetype},
 and two types of unbounded edges.
\end{lemma}
\begin{proof}
If $\theta$ is regular, then there is no $4$-square in orientation $\theta$.
So, all vertices $v \in V(\theta) \setminus (P \cup \{\infty\})$ are regular.
Since there is no vertex $v$ with stapled contact type,
every edge in $E(\theta)$ is also regular by Lemma~\ref{lem:vertextype}.
This proves the first statement of the lemma.

There are two types of regular vertices, namely, $(3,2)$-type and $(3,3)$-type,
as shown in \figurename~\ref{fig:vertextype}.
There are four different types of edges between two regular vertices:
$(3,2)$--$(3,2)$ sliding, $(3,2)$--$(3,3)$ sliding, $(3,2)$--$(3,3)$ growing,
and $(3,3)$--$(3,3)$ sliding.
If an edge is incident to a point $p\in P \subset V(\theta)$,
then it is growing and the other endpoint of it is always a vertex
of $(3,2)$-type.
This is the fifth type for bounded edges in $E(\theta)$.
Hence, there are five types for bounded edges
as shown in \figurename~\ref{fig:edgetype}.
For unbounded edges, only two types are possible for their incident
vertices, namely, $(3,2)$-type and $(4,1)$-type,
by Lemma~\ref{lem:vertextype} and \figurename~\ref{fig:vertextype}.
\end{proof}

%
%

Consider any contact type $\ct$ with three contact pairs
and three pinned sides.
For any $\theta \in \OS$,
let $S_\ct(\theta)$ be the square in orientation $\theta$
such that $\ct$ is a subset of its contact type, regardless of its emptiness.
Note that $S_\ct(\theta)$ cannot be well defined for all $\theta \in \OS$,
whereas it is well defined in one closed interval of $\OS$
or is not defined for all $\theta \in \OS$.
For each $\theta$ for which $S_\ct(\theta)$ is well defined,
we call $\theta$ \emph{valid for $\ct$} if $S_\ct(\theta)$ is empty
and its contact type is exactly $\ct$.
Note that if $\theta$ is valid for $\ct$, then
there exists a vertex $v\in V(\theta)$ with $\ct_v = \ct$.
\begin{lemma} \label{lem:valid_interval}
 For any contact type $\ct$ with three contact pairs and three pinned sides,
 the set of all valid orientations for $\ct$ forms
 zero or more open intervals $I \subset \OS$ such that
 $S_\ct(\phi)$ is a $4$-square for any endpoint $\phi$ of $I$.
\end{lemma}
\begin{proof}
If there is no valid orientation for $\ct$, then we are done.
Suppose that some $\theta_0\in\OS$ is valid for $\ct$.
Then, the square $S_\ct(\theta_0)$ is of contact type $\ct$ and is empty.
Consider the rotation of the square in one direction,
that is, $S_\ct(\theta)$ as $\theta$ continuously increases
from $\theta = \theta_0$.
It is obvious that $S_\ct(\theta)$ is still empty and has contact type $\ct$
for any $\theta$ sufficiently close to $\theta_0$,
so $\theta$ locally near $\theta_0$ is valid for $\ct$.
We continue increasing $\theta$ until we reach some $\phi\in\OS$
that is not valid for $\ct$.
This is one of the following cases:
a contact point in $\ct$ reaches a corner of $S_\ct(\phi)$
or another point in $P$ is hit by the boundary of $S_\ct(\phi)$.
In either case, it gains one more contact pair and thus
$S_\ct(\phi)$ is a $4$-square.
Note that $\phi$ is not valid for $\ct$.

It is symmetric to rotate $S_\ct(\theta)$ in the opposite direction
by decreasing $\theta$.
Thus, any valid orientation $\theta_0$ for $\ct$ lies in between
two orientations that are not valid for $\ct$
such that any $\theta$ in between them is valid for $\ct$.
This implies that the set of all valid orientations for $\ct$
indeed forms zero or more open intervals $I$ in $\OS$.
For any endpoint $\phi$ of $I$, $S_\ct(\phi)$ is a $4$-square
as observed above.
Therefore, the lemma is shown.
\end{proof}
We call each of these open intervals described in Lemma~\ref{lem:valid_interval}
a \emph{valid interval} for $\ct$.
Lemma~\ref{lem:valid_interval} also implies that
any endpoint of a valid interval for $\ct$ is a degenerate orientation.

\begin{figure}[htbp]
\begin{center}
\includegraphics[width=.80\textwidth]{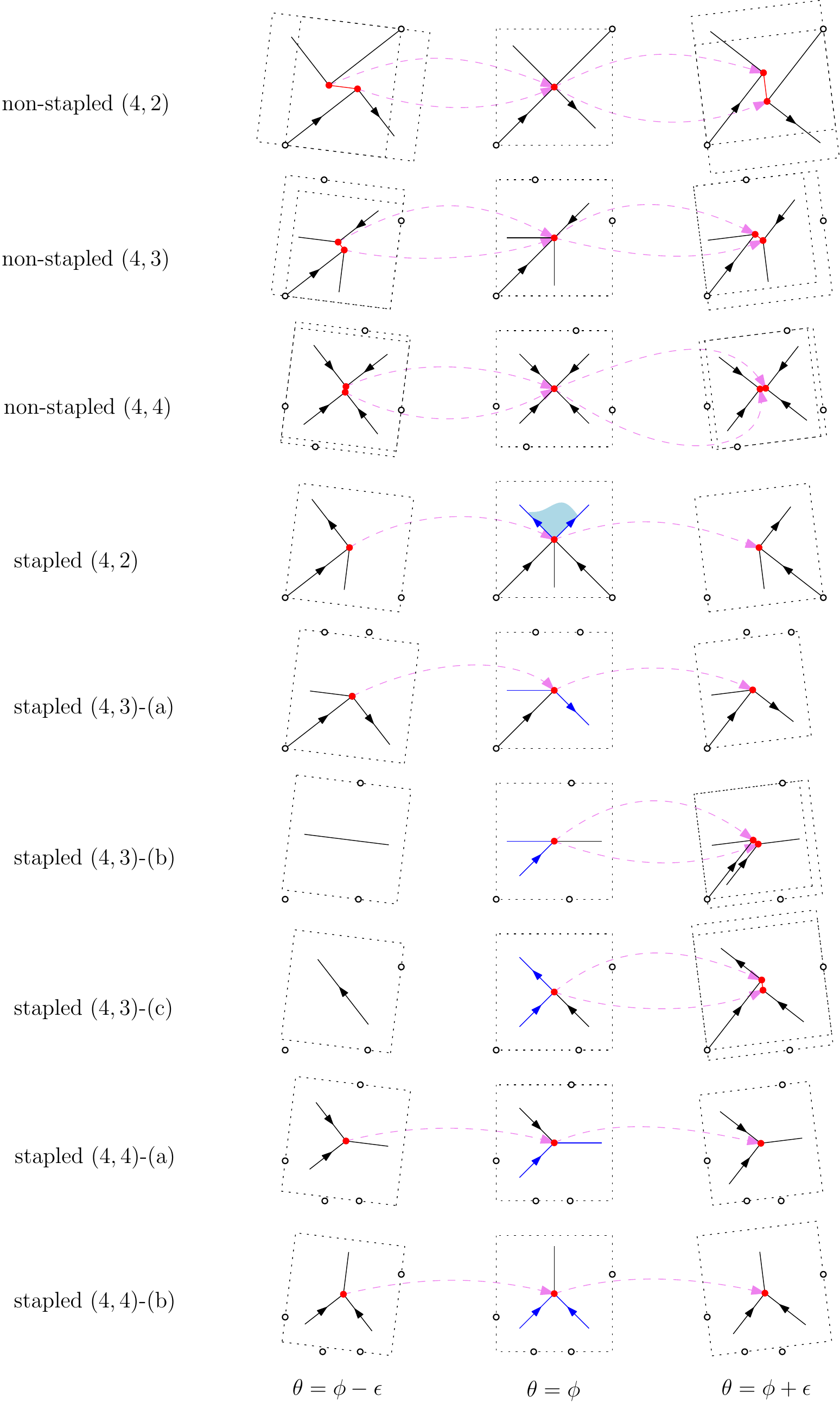}
\end{center}
\caption{Transitions between regular ($3$-squares) and non-regular vertices
 ($4$-squares).
 }
\label{fig:transition}
\end{figure}

For any degenerate orientation $\phi \in \OS$,
let $\Sq_4(\phi) \subseteq \Sq_4$ be the set of $4$-squares
whose orientation is $\phi$.
Note that $\Sq_4(\phi)$ is nonempty and consists of at most $O(n)$ squares
by Lemma~\ref{lem:VD_complexity}.
We then observe the following.
\begin{lemma} \label{lem:4sq-3sq}
 Let $S\in\Sq_4(\phi)$ be any $4$-square in orientation $\phi$
 and $\ct$ be its contact type.
 Then, there are exactly two or four contact types $\ct'$
 with three contact pairs and three pinned sides
 such that $S_{\ct'}(\phi) = S$ and
 either $\phi-\epsilon$ or $\phi+\epsilon$ is valid for $\ct'$
 for any arbitrarily small $\epsilon>0$.
 Specifically, the number of such contact types $\ct'$
 is four if $S$ is non-stapled, or two
 if $S$ is stapled.
\end{lemma}
\begin{proof}
If $S_{\ct'}(\phi) = S$, then the contact type of $S_{\ct'}(\phi)$
is exactly $\ct$ and it includes the three contact pairs in $\ct'$,
so we have $\ct' \subset \ct$.
Since $|\ct| = 4$, the number of such $\ct'$ is at most four.

In order to prove the lemma,
we handle every possible types for $S$ separately.
There are nine possible types for $S$ as shown in \figurename~\ref{fig:4sq_type},
except the trivial $(4,1)$ type.
For each case, we observe that there are
exactly two or four contact types $\ct'$ with the stated property
by a local transformation from $S$.
See \figurename~\ref{fig:transition}.
In the following, $\epsilon$ denotes any arbitrarily small positive real number.

First, suppose that $S\in\Sq_4(\phi)$ is a non-stapled $4$-square.
Then, every side of $S$ is pinned by its contact points,
so we can write $\ct = \{(p_1,\mt), (p_2,\mr), (p_3,\mb), (p_4,\ml)\}$
for $p_1,p_2,p_3,p_4 \in P$.
There are three different non-stapled types for $S$.
\begin{itemize} \denseitems
 \item If $S$ is a $(4,2)$-square,
 then either (1) $p_1 = p_2$ and $p_3 = p_4$, or (2) $p_1=p_4$ and $p_2=p_3$.
 Without loss of generalization, we assume the former case.
 The other case is symmetric.
 Consider all four subsets of $\ct$ with three contact pairs, namely,
 $\ct_u = \{(p_1,\mt), (p_1,\mr), (p_3,\mb)\}$,
 $\ct_v = \{(p_1,\mt), (p_3,\ml), (p_3,\mb)\}$,
 $\ct_{u'} = \{(p_1,\mt), (p_1,\mr), (p_3,\ml)\}$, and
 $\ct_{v'} = \{(p_1,\mr), (p_3,\ml), (p_3,\mb)\}$.
 We then observe that $\phi-\epsilon$ is valid for $\ct_u$ and $\ct_v$
 and $\phi+\epsilon$ is valid for $\ct_{u'}$ and $\ct_{v'}$.
 Hence, $u,v\in V(\phi-\epsilon)$ and $u', v' \in V(\phi+\epsilon)$.
 See the first row of \figurename~\ref{fig:transition}.
 Also, there is a sliding edge $e \in E(\phi-\epsilon)$ between $u$ and $v$ with
 $\ct_e = \ct_u \cap \ct_v = \{(p_1,\mt), (p_3,\mb)\}$;
 a sliding edge $e' \in E(\phi+\epsilon)$ between $u'$ and $v'$
 with $\ct_{e'} = \ct_{u'}\cap \ct_{v'} = \{(p_1,\mr), (p_3,\ml)\}$.
 \item If $S$ is a $(4,3)$-square,
 then we have four cases:
 (1) $p_1=p_2$, (2) $p_2=p_3$, (3) $p_3=p_4$, or (4) $p_4=p_1$,
 depending on which corner of $S$ has a contact point.
 Assume that $p_3 = p_4$, as the other cases are symmetric.
 Consider all four subsets of $\ct$ with three contact pairs, namely,
 $\ct_u = \{(p_1,\mt), (p_2,\mr), (p_3,\mb)\}$,
 $\ct_v = \{(p_2,\mr), (p_3,\mb), (p_3,\ml)\}$,
 $\ct_{u'} = \{(p_1,\mt), (p_2,\mr), (p_3,\ml)\}$, and
 $\ct_{v'} = \{(p_1,\mt), (p_3,\mb), (p_3,\ml)\}$.
 We then observe that $\phi-\epsilon$ is valid for $\ct_u$ and $\ct_v$
 and $\phi+\epsilon$ is valid for $\ct_{u'}$ and $\ct_{v'}$.
 See the second row of \figurename~\ref{fig:transition}.
 Also, there is a growing edge $e \in E(\phi-\epsilon)$ between $u$ and $v$ with
 $\ct_e = \ct_u \cap \ct_v = \{(p_2,\mr), (p_3,\mb)\}$;
 there is a growing edge $e' \in E(\phi+\epsilon)$ between $u'$ and $v'$
 with $\ct_{e'} = \ct_{u'}\cap \ct_{v'} = \{(p_1,\mt), (p_3,\ml)\}$.
 \item If $S$ is a $(4,4)$-square,
 then $p_1, \ldots, p_4$ are distinct.
 As above, we check all four subsets of $\ct$ with three contact pairs,
 and observe that $\phi-\epsilon$ is valid for two of them
 and $\phi+\epsilon$ is valid for the other two.
 See the third row of \figurename~\ref{fig:transition}.
\end{itemize}
Consequently, if $S$ is non-stapled,
then all the four subsets $\ct'\subset \ct$ of three contact pairs
satisfy the stated condition that 
$S_{\ct'}(\phi) = S$ and either $\phi-\epsilon$ or $\phi+\epsilon$
is valid for $\ct'$.
More precisely, since each such $\ct'$ determines a regular vertex,
two regular vertices in $V(\phi-\epsilon)$ are indeed merged into
the non-regular vertex corresponding to $S$ in $V(\phi)$
and then it splits into two other regular vertices in $V(\phi+\epsilon)$.
This indeed describes an \emph{edge flip} of the Voronoi diagram $\VD(\theta)$
as $\theta$ goes through $\phi$.

Next, we suppose that $S\in\Sq_4(\phi)$ is a stapled $(4,2)$-square,
and assume that its contact type $\ct = \{(p,\ml),(p,\mb),(q,\mb),(q,\mr)\}$
for $p,q\in P$.
See the fourth row of \figurename~\ref{fig:transition}.
The other cases are symmetric.
Checking all possible subsets of $\ct$ with three elements,
we observe that $\phi-\epsilon$ is valid only for 
$\ct_v=\{(p,\ml),(p,\mb),(q,\mr)\}$
and $\phi+\epsilon$ is valid only for 
$\ct_{v'}=\{(p,\ml),(q,\ml),(q,\mr)\}$.

Third, we suppose that $S$ is a stapled $(4,3)$-square.
As seen in \figurename~\ref{fig:4sq_type},
exactly one contact point of $S$ lies on a corner of $S$.
Without loss of generality, assume that $p\in P$ lies on
the bottom-left corner of $S$,
so $(p,\mb)$ and $(p,\ml)$ belong to its contact type $\ct$.
Let $q, r\in P$ be the other two contact points of $S$.
\begin{itemize} \denseitems
 \item If $S$ is of stapled $(4,3)$-(a)-type, then
 we assume that both $q$ and $r$ lie on the top side of $S$
 and $q$ is to the left of $r$ in orientation $\phi$.
 So, we have $\ct = \{(p,\mb), (p,\ml), (q,\mt), (r,\mt)\}$.
 See the fifth row of \figurename~\ref{fig:transition}.
 Similarly, checking all possible subsets of $\ct$ with three elements,
 we observe that $\phi-\epsilon$ is valid only for 
 $\ct_v=\{(p,\ml),(p,\mb),(q,\mt)\}$
 and $\phi+\epsilon$ is valid only for 
 $\ct_{v'}=\{(p,\ml),(p,\mb),(r,\mt)\}$.
 \item If $S$ is of stapled $(4,3)$-(b)-type, then
 we assume that $q$ lies on the bottom side of $S$
 and $r$ lies on the top side of $S$.
 So, we have $\ct = \{(p,\mb), (p,\ml), (q,\mb), (r,\mt)\}$.
 See the sixth row of \figurename~\ref{fig:transition}.
 In this case,
 we observe that $\phi+\epsilon$ is valid only for 
 $\ct_{u'}=\{(p,\ml),(q,\mb),(r,\mt)\}$ and
 $\ct_{v'}=\{(p,\ml),(p,\mb),(r,\mt)\}$,
 and that $\phi-\epsilon$ is valid for none of the four subsets
 of $\ct$ with three elements.
 Note that $u', v'\in V(\phi+\epsilon)$ and there is a growing edge
 $e'\in E(\phi+\epsilon)$ such that 
 $\ct_{e'} = \ct_{u'}\cap \ct_{v'} = \{(p,\ml),(r,\mt)\}$.
 Also, notice that the above case is one of the other symmetric cases
 of the stapled $(4,3)$-(b)-type,
 hence we may also have the reversed situation:
 $\phi-\epsilon$ is valid for two contact types that are subsets of $\ct$
 and $\phi+\epsilon$ is not for any of them.
 \item If $S$ is of stapled $(4,3)$-(c)-type, then
 we assume that $q$ lies on the bottom side of $S$
 and $r$ lies on the right side of $S$.
 So, we have $\ct = \{(p,\mb), (p,\ml), (q,\mb), (r,\mr)\}$.
 See the seventh row of \figurename~\ref{fig:transition}.
 In this case,
 we observe that $\phi+\epsilon$ is valid only for 
 $\ct_{u'}=\{(p,\ml),(q,\mb),(r,\mr)\}$ and
 $\ct_{v'}=\{(p,\ml),(p,\mb),(r,\mr)\}$,
 and that $\phi-\epsilon$ is valid for none of the four subsets
 of $\ct$ with three elements.
 Similarly to the stapled $(4,3)$-(c)-type,
 $u', v'\in V(\phi+\epsilon)$ and there is a growing edge
 $e'\in E(\phi+\epsilon)$ such that 
 $\ct_{e'} = \ct_{u'}\cap \ct_{v'} = \{(p,\ml),(r,\mr)\}$.
 Also, in a symmetric configuration, we may have the reversed situation:
 $\phi-\epsilon$ is valid for two contact types that are subsets of $\ct$
 and $\phi+\epsilon$ is not for any of them.
\end{itemize}
Hence, if $S$ is of a stapled $(4,3)$-square, then
there are exactly two subsets $\ct'\subset\ct$ of three contact pairs
such that either $\phi-\epsilon$ or $\phi+\epsilon$ is valid for $\ct'$
and $S_{\ct'}(\phi) = S$.

Finally, suppose that $S$ is a stapled $(4,4)$-square.
There are two contact points $p,q\in P$ of $S$ that lie on a common side of $S$.
Without loss of generality, assume that 
$p$ and $q$ lie on the bottom side of $S$, and that
$p$ is to the left of $q$ in orientation $\phi$.
So, $(p,\mb)$ and $(q,\mb)$ belong to its contact type $\ct$.
We have two cases.
\begin{itemize} \denseitems
 \item If $S$ is of stapled $(4,4)$-(a)-type, 
 then a third contact point $r_1\in P$ lies on the top side of $S$,
 and the last contact point $r_2\in P$ lies
 either on the left or right side of $S$.
 Assume that $r_2$ lies on the left side of $S$,
 so $\ct = \{(p,\mb),(q,\mb),(r_1,\mt),(r_2,\ml)\}$.
 See the eighth row of \figurename~\ref{fig:transition}.
 The other case is symmetric.
 As done above, consider all all possible subsets of $\ct$ with three elements
 and observe that
 $\phi-\epsilon$ is valid only for
 $\ct_v=\{(q,\mb),(r_1,\mt), (r_2,\ml)\}$ and
 $\phi+\epsilon$ is valid only for
 $\ct_v=\{(p,\mb),(r_1,\mt), (r_2,\ml)\}$.
 \item If $S$ is of stapled $(4,4)$-(b)-type, 
 then each of the other two contact points $r_1, r_2\in P$
 lies on the left and the right side of $S$, respectively.
 So, we have $\ct = \{(p,\mb),(q,\mb),(r_1,\mr),(r_2,\ml)\}$.
 See the ninth row of \figurename~\ref{fig:transition}.
 We then observe that
 $\phi-\epsilon$ is valid only for
 $\ct_v=\{(q,\mb),(r_1,\mr), (r_2,\ml)\}$ and
 $\phi+\epsilon$ is valid only for
 $\ct_v=\{(p,\mb),(r_1,\mr), (r_2,\ml)\}$.
\end{itemize}
In either case, there are exactly two subsets $\ct'\subset\ct$ 
of three contact pairs
such that either $\phi-\epsilon$ or $\phi+\epsilon$ is valid for $\ct'$
and $S_{\ct'}(\phi) = S$.

Summarizing, if $S$ is a stapled $4$-square,
then exactly two subsets $\ct'\in\ct$ satisfy the stated condition 
of the lemma.
So, the lemma is shown.
\end{proof}
\figurename~\ref{fig:transition} illustrates transitions
of a non-regular vertex, which corresponds to a $4$-square in $\Sq_4(\phi)$,
from and to regular vertices which correspond to $3$-squares,
locally at $\theta=\phi$,
so almost describes our proof for Lemma~\ref{lem:4sq-3sq}.
Observe that any non-stapled $4$-square is relevant to an \emph{edge flip}
of $\VD(\theta)$, and that
only stapled $(4,3)$-(b)- or (c)-type squares show a bit different behavior;
such a non-regular vertex suddenly appears on a regular edge
and splits to two regular vertices,
or, in the opposite way,
two regular vertices are merged into such a non-regular one
and soon disappear.

The above discussions provide us a thorough view on the relation
among the $3$-squares and the $4$-squares.
Consider the graph $\GSq$ whose vertex set is the set $\Sq_4$ of all $4$-squares
and whose edge set consists of edges $(S, S')$ for $S, S'\in\Sq_4$
such that there is a valid interval $I=(\phi,\phi')$ for some
contact type $\ct$ such that $S_\ct(\phi) = S$ and $S_{\ct}(\phi') = S'$.
The graph $\GSq$ is well defined by Lemma~\ref{lem:valid_interval}
and the degree of every vertex in $\GSq$ is
exactly two or four, depending on its type, by Lemma~\ref{lem:4sq-3sq}.
The very natural embedding of $\GSq$ is on the three-dimensional
space $\Plane\times\OS$
in such a way that each vertex $S\in\Sq_4$ is put on a point $(c, \phi)$
where $c$ and $\phi$ denote the center and the orientation of $S$,
and each edge is drawn by the locus of the corresponding regular vertex
in $\VD(\theta)$ for $\theta$ lying in its valid interval.

\section{Number of $4$-Squares} \label{sec:4sq}

In this section, we prove asymptotically tight upper and lower bounds
on the number of $4$-squares and $(4,k)$-squares for each $2\leq k \leq 4$.
Following are two main theorems.
\begin{theorem} \label{thm:4sq}
 The number of $4$-squares among $n$ points in general position is
 always between $\Omega(n)$ and $O(n^2)$.
 These lower and upper bounds are asymptotically tight.
\end{theorem}

\begin{theorem} \label{thm:44sq}
 Among $n$ points in general position,
 the number of empty squares whose boundary contains four points of $P$
 is between $0$ and $O(n^2)$.
 These lower and upper bounds are asymptotically tight.
\end{theorem}

For any positive integers $m$ and $k \leq m$,
let $s_m(P)$ and $s_{m,k}(P)$ be the number of $m$-squares and $(m,k)$-squares,
respectively, among $P$.
We then define
\begin{align*}
  \sigma_m(n) &:= \min_{|P|=n} s_m(P),
  & \sigma_{m,k}(n) &:= \min_{|P|=n} s_{m,k}(P), \\
  \Sigma_m(n) &:= \max_{|P|=n} s_m(P), \qquad \text{ and}
  & \Sigma_{m,k}(n) &:=  \max_{|P|=n} s_{m,k}(P).
\end{align*}

Note that $\sigma_m(n) = \Sigma_m(n) = 0$ for any $m > 4$
by our general position assumption, and
$\sigma_m(n)$ and $\Sigma_m(n)$ are unbounded for $m < 4$.
Some immediate bounds are:
$s_m(P) = \sum_{1\leq k \leq m} s_{m,k}(P)$,
$\Sigma_m(n) \leq \sum_{1\leq k \leq m} \Sigma_{m,k}(n)$,
and $\sigma_m(n) \geq \sum_{1\leq k \leq m} \sigma_{m,k}(n)$.

In the following, we show the asymptotically tight bounds on these quantities
for $m=4$.
It is obvious that
$\sigma_{4,1}(n) = \Sigma_{4,1}(n) = n$ and
$\Sigma_{4,k}(n) \leq {n \choose k}$ for $2\leq k \leq 4$.
So, we have $\sigma_4(n) = \Omega(n)$ and $\Sigma_4(n) = O(n^4)$.
In this paper, we are interested  only in their asymptotic bounds,
while, however, finding the exact constants hidden in the analysis would be
another interesting combinatorial problem.
%

\subsection{Upper bounds}
Here, we prove the upper bounds $\Sigma_{4,k}(n)$ for $2\leq k\leq 4$
and $\Sigma_4(n)$ are quadratic in $n$.
We first show that there exist a family of point sets $P$ with $n = |P|$
having this many $4$-squares.

\begin{figure}[tbh]
\begin{center}
\includegraphics[width=.3\textwidth]{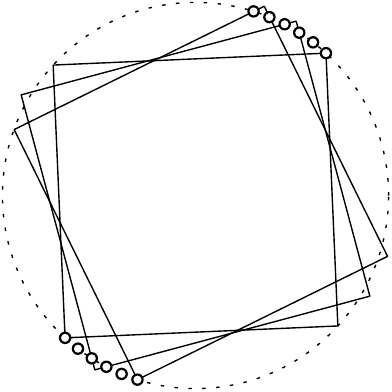}
\end{center}
\caption{Illustration of a point set $P_n$ with $\Omega(n^2)$ many $4$-squares.}
\label{fig:4sq_upper_example}
\end{figure}

\begin{lemma}\label{lem:4sq_upper_example}
 For any $n \geq 4$, there exists a set $P_n$ of $n$ points
 such that
 $s_{4,k}(P) = \Omega(n^2)$ for each $2\leq k\leq 4$
 and thus $s_4(P) = \Omega(n^2)$.
\end{lemma}
\begin{proof}
Let $n \geq 4$ be an integer.
We first describe how to construct a point set $P_n$.
We assume that $n$ is even.
Let $\epsilon \in \Real$ be a sufficiently small positive real number,
and $\delta := 2\epsilon / n $.
Consider a unit circle $C$ centered at the origin $o\in \Plane$.
For each $i \in \{1, \ldots, n/2\}$,
let $p^+_i$ and $p^-_i$ be the two intersection points between $C$ and
the line $\{y = (1+i\delta)x\}$ such that $p^+_i$ is the one with positive coordinates 
and $p^-_i$ is the other.
Let $P^+ := \{p^+_1, \ldots, p^+_{n/2}\}$ and $P^-:=\{p^-_1,\ldots, p^-_{n/2}\}$.
Our point set $P_n$ is defined to be $P_n := P^+ \cup P^-$.
See Figure~\ref{fig:4sq_upper_example}.

Pick any $1\leq i, j \leq n/2-1$.
It is not difficult to see that there exists a non-stapled $(4,2)$-square $S_{ij}$
with contact type $\{(p^+_i,\mt),(p^+_i,\mr),(p^-_j,\mb),(p^-_j,\ml)\}$.
This proves that $s_{4,2}(P_n) = \Omega(n^2)$.

Let $\phi\in\OS$ be the orientation of $S_{ij}$.
We consider the square
$S_\ct(\theta)$ where $\ct = \{(p^+_i,\mr),(p^-_j,\mb),(p^-_j,\ml)\}$
for $\theta \geq \phi$.
Note that $S_\ct(\phi) = S_{ij}$.
We increase $\theta$ from $\phi$ until it hits another point $p\in P_n$ at $\phi' > \phi$.
Observe that $p = p^+_{i+1}$ by our construction of $P_n$.
Hence, $S_\ct(\phi')$ is a non-stapled $(4,3)$-square
with contact type $\{(p^+_{i+1},\mt),(p^+_i,\mr),(p^-_j,\mb),(p^-_j,\ml)\}$.
This proves that $s_{4,3}(P_n) = \Omega(n^2)$.

Finally, consider the square $S_{\ct'}(\theta)$
where $\ct' = \{(p^+_{i+1},\mt),(p^+_i,\mr),(p^-_j,\ml)\}$.
We again increase $\theta$ from $\phi'$ until $S_{\ct'}(\theta)$ hits 
a fourth point $p'\in P_n$ at $\phi'' > \phi'$.
Observe that $p' = p^-_{j+1}$ by our construction of $P_n$.
Hence, $S_{\ct'}(\phi'')$ is a non-stapled $(4,4)$-square
with contact type $\{(p^+_{i+1},\mt),(p^+_i,\mr),(p^-_{j+1},\mb),(p^-_j,\ml)\}$.
This proves that $s_{4,4}(P_n) = \Omega(n^2)$.

Therefore, we have $s_4(P_n) = \sum_{1\leq k \leq 4}$ and 
$s_{4,k}(P_n) = \Omega(n^2)$,
so the lemma is proven.
%
\end{proof}

This already proves that $\Sigma_{4,2}(n) = \Theta(n^2)$.
In the following, we prove the matching upper bounds of
$\Sigma_{4,3}(n)$ and $\Sigma_{4,4}(n)$.

\paragraph*{Upper bound of $\Sigma_{4,3}(n)$}
Any $(4,3)$-square is of one of the four types:
non-stapled $(4,3)$ and stapled $(4,3)$-(a--c) types,
see \figurename~\ref{fig:4sq_type}.
We bound the number of $(4,3)$-squares of each type, separately.

First, consider the stapled $(4,3)$-(b--c) types.
Note that if $S$ is in this case with two contact points $p,q\in P$
on its stapled side,
then one of $p$ and $q$ also lie on a corner of $S$.
\begin{lemma} \label{lem:43bc-42}
 Let $S \in \Sq_4(\phi)$ be a square of stapled $(4,3)$-(b)- or $(4,3)$-(c)-type
 with contact type $\ct$,
 and $v\in V(\phi)$ be the vertex with $\ct_v = \ct$.
 Then, there exists a stapled $(4,2)$-square $S'\in\Sq_4(\phi)$
 with contact type $\ct'$
 such that $v$ is adjacent to $v'\in V(\phi)$ with $\ct_{v'} = \ct'$
 in $\VG(\phi)$.
\end{lemma}
\begin{proof}
We shrink the square $S$ towards the corner on which a contact point lies.
Without loss of generality, we assume that $p\in P$ lies on the bottom-left corner
of $S$, and another contact point $q\in P$ lies on the bottom side of $S$.
We thus have $\ct = \{(p,\mb),(p,\ml),(q,\mb), (p',\ma)\}$
for some $p'\in P \setminus\{p,q\}$ and  some $\ma \in \{\mt, \mr\}$.
Note that if $\ma = \mt$, then $S$ is of stapled $(4,3)$-(b)-type;
if $\ma = \mr$, then $S$ is of stapled $(4,3)$-(b)-type.
See \figurename~\ref{fig:vertextype}.

Consider the largest empty square $S(c)$ centered at $c\in\Plane$
in orientation $\phi$.
Notice that there is a growing edge $e =vp \in E(\phi)$ between $v$ and $p$
with $\ct_e = \{(p,\mb),(p,\ml),(q,\mb)\}$ and its growing direction
is towards $v$.
We shrink $S$ by considering $S(c)$
as $c$ continuously moves from $\hat{v}$ along the edge $\hat{e}$.
Note that $S = S(\hat{v})$ and the contact type of $S(c)$ is $\ct_e$
if $c$ lies in the relative interior of $\hat{e}$.
We stop $c$ when $S(c)$ gains a fourth contact pair, so $c$ reaches
the other vertex $v'\in V(\phi)$ incident to $e$.
By Lemma~\ref{lem:edge_vertex},
we have $S(\hat{v'}) \subset S=S(\hat{v})$,
so the new contact pair does not involve points other than $p$ and $q$.
Hence, $\ct_{v'} = \{(p,\mb),(p,\ml),(q,\mb),(q,\mr)\}$
and thus $S(\hat{v'})$ is a stapled $(4,2)$-square.
\end{proof}
From Lemma~\ref{lem:vertextype}, we know that each vertex $v'\in V(\phi)$
of stapled $(4,2)$-type has two growing edges whose growing direction
is outwards from $v'$.
Hence, for each vertex of stapled $(4,2)$-type,
there can be at most two adjacent vertices whose corresponding square
is larger and has three contact points.
This, together with Lemma~\ref{lem:43bc-42}, implies that
the number of stapled $(4,3)$-(b)- and (c)-type squares
is at most twice the number of stapled $(4,2)$-squares,
which is $O(n^2)$.

Next, we consider the other two types of $(4,3)$-squares:
non-stapled $(4,3)$-type and stapled $(4,3)$-(a)-type.
Consider any $(4,3)$-square $S\in\Sq_4(\phi)$
whose type is one of the above two.
Without loss of generality, assume that
a contact point $p\in P$ lies on the bottom-left corner of $S$.
Then, regardless of its specific type,
the other two contact points $q_1,q_2\in P$ lie on either
the top or the right side of $S$.
So, $\ct = \{(p,\mb),(p,\ml),(q_1,\ma_1),(q_2,\ma_2)\}$
for $\ma_1,\ma_2\in\{\mt,\mr\}$.
See \figurename~\ref{fig:4sq_type}.
Notice that $q_1$ and $q_2$ are two equidistantly closest points from $p$
under the distance function $d_\phi$
among those points in the quadrant with apex $p$ in orientation $\phi$.

In the following, we count the number of those pairs $(q_1, q_2)$
with the discussed property for each fixed $p\in P$.
This bounds the number of those $(4,3)$-squares with $p$
on the bottom-left corner.
More precisely,
let $\ell_\theta$ be the upward half-line from $p$
in orientation $\theta\in [0,\pi)$,
and $\cone_\theta$ be the cone with apex $p$ and angle span $\pi/4$
defined by two half-lines $\ell_\theta$ and $\ell_{\theta+\pi/4}$.
Then, for any $\theta \in \OS$,
consider the two subsets
\[Q_\mr(\theta):= P\setminus\{p\} \cap \cone_\theta \quad \text{ and }
\quad Q_\mt(\theta):= P\setminus\{p\}\cap \cone_{\theta+\pi/4}.\]
For any $q \in P\setminus\{p\}$, define a function $f_q \colon \OS \to \Real$
to be
\[ f_q(\theta) := \begin{cases}
          	d_\theta(p,q) & \text{if } q\in Q_\mr(\theta)\cup Q_\mt(\theta) \\
          	\infty & \text{otherwise}
          		  \end{cases}.
\]
Then, our task is reduced to find the complexity of
the lower envelope $F$ of the functions $f_q$ for all $q\in P\setminus\{p\}$.
We divide the quadrant with apex $p$ into two cones $\cone_\theta$ and
$\cone_{\theta+\pi/4}$ since $d_\theta(p,q)$ is represented by
a sinusoidal function in one of the two cones.
Specifically, we have
$d_\theta(p,q) = |pq|\cdot\cos(\alpha-\theta)$
if $q\in Q_\mr(\theta)$;
and $d_\theta(p,q) = |pq|\cdot\sin(\alpha-\theta)$,
where $|pq|$ denotes the Euclidean distance between $p$ and $q$
and $\alpha$ denotes the orientation of segment $pq$.

Hence, each $f_q$ consists of exactly two sinusoidal functions
of the same period $2\pi$.
Since any two such sinusoidal curves can cross at most once in domain
$\OS=[0,\pi/2)$,
The complexity of their lower envelope $F(\theta) = \min_{q} f_q(\theta)$
is at most $O(n\alpha(n))$
by the theory of Davenport--Schinzel sequences~\cite{sa-dsstga-95}.
Fortunately, the domains on which the curves in our hand are defined
have a certain structure
so that it suffices to reduce the complexity bound down to $O(n)$
by a trick similar to that
which has been used in Hershberger~\cite{h-fuenls-89}.
\begin{lemma}\label{lem:43_upper_F}
 The complexity of the lower envelope $F$ of functions $f_q$
 for all $q\in P\setminus\{p\}$ is bounded by $O(n)$.
\end{lemma}
\begin{proof}
Consider the graph of $f_q$ drawn in space $\OS\times \Real$.
As discussed above, the graph of $f_q$ consists of
at most two sinusoidal curves.
Let $\Gamma$ be the union of all these curves.
Note that $\Gamma$ consists of at most $2n$ curves,
any two of which cross at most once.
For each $\gamma \in \Gamma$, let $D_\gamma \subset \OS$ be
the domain of the corresponding partial sinusoidal function.

Consider any $\gamma\in\Gamma$.
Note that the length of interval $D_\gamma \subset \OS$ is at most $\pi/4$
by definition.
Thus, this is exactly one of three cases:
either $0\in D_\gamma$, $\pi/4 \in D_\gamma$, or $\pi/2-\epsilon\in D_\gamma$
for any arbitrarily small positive $\epsilon>0$.
Let $\Gamma_0, \Gamma_{\pi/4}, \Gamma_{\pi/2} \subseteq \Gamma$
be the sets of those curves in $\Gamma$
whose domain contains $0$, $\pi/4$, and $\pi/2-\epsilon$, respectively.
Then, the three sets form a disjoint partition of
$\Gamma=\Gamma_0 \cup \Gamma_{\pi/4}\cup \Gamma_{\pi/2}$.

Let $F_\theta$ be the lower envelope of curves in $\Gamma_\theta$
for $\theta \in \{0, \pi/4, \pi/2\}$.
Since all the curves in $\Gamma_0$ start at the same $\theta = 0$
and any two of them cross at most once,
their lower envelope $F_0$ corresponds to
the Davenport--Schinzel sequence of order two,
so the complexity of $F_0$ is bounded by $O(n)$.
The same argument can be found in Hershberger~\cite[Lemma 3.1]{h-fuenls-89}.
Similarly, the complexity of $F_{\pi/4}$ and $F_{\pi/2}$ is also $O(n)$.
In particular, for $F_{\pi/4}$, we cut each curve in $\Gamma_{\pi/4}$
at $\theta = \pi/4$ and separately consider those to the left and to the right
of $\theta = \pi/4$.

Finally, observe that the lower envelope $F$ of functions $f_q$
corresponds to the lower envelope of $F_0$, $F_{\pi/4}$, and $F_{\pi/2}$.
Since there are at most three curves in $F_0 \cup F_{\pi/4} \cup F_{\pi/2}$
defined at every $\theta \in \OS$,
the complexity of $F$ remains $O(n)$.
\end{proof}
This proves that the number of non-stapled $(4,3)$-squares
and stapled $(4,3)$-(a)-squares among $P$
is at most $O(n^2)$.
Combining the above discussions about stapled $(4,3)$-(b--c) squares
and Lemma~\ref{lem:4sq_upper_example},
we conclude that $\Sigma_{4,3}(n) = \Theta(n^2)$.

\paragraph*{Upper bound of $\Sigma_{4,4}(n)$}
Now, we prove the matching upper bound on the number of $(4,4)$-squares.
Any $(4,4)$-square is one of the three types:
non-stapled $(4,4)$ and stapled $(4,4)$-(a--b) types.
See \figurename~\ref{fig:4sq_type}.

\begin{lemma} \label{lem:44-43}
 Let $S \in \Sq_4(\phi)$ be a stapled $(4,4)$-square with contact type $\ct$,
 and $v\in V(\phi)$ be the vertex with $\ct_v = \ct$.
 Then, there exists a stapled $(4,3)$-square $S'\in\Sq_4(\phi)$
 with contact type $\ct'$
 such that $v$ is adjacent to $v'\in V(\phi)$ in $\VG(\phi)$
 with $\ct_{v'} = \ct'$.
\end{lemma}
\begin{proof}
There are two possible types for $S$:
stapled $(4,4)$-(a) and stapled $(4,4)$-(b).
Without loss of generality, we assume that the bottom side of $S$ is
stapled by two contact points $p,q\in P$
and the left side is pinned by a third contact point $r\in P$.
We thus have $\ct = \{(p,\mb),(q,\mb),(r,\ml), (p',\ma)\}$
for some $p'\in P \setminus\{p,q,r\}$ and  some $\ma \in \{\mt, \mr\}$.
Note that if $\ma = \mt$, then $S$ is of stapled $(4,4)$-(a)-type;
if $\ma = \mr$, then $S$ is of stapled $(4,4)$-(b)-type.
See \figurename~\ref{fig:vertextype}.

Let $S(c)$ be the maximal empty square centered at $c\in\Plane$
in orientation $\phi$.
Notice that there is a growing edge $e\in E(\phi)$ incident to $v$
with $\ct_e = \{(p,\mb),(q,\mb),(r,\ml)\}$,
regardless of the type of $S$.
Consider the motion of $S(c)$
as $c$ continuously moves from $\hat{v}$ along the edge $\hat{e}$.
Note that $S = S(\hat{v})$ and the contact type of $S(c)$ is $\ct_e$
in the relative interior of $\hat{e}$.
Since $e$ is a growing edge directed towards $v$,
$S(c)$ is shrinking as $c$ moves along $\hat{e}$.
We stop $c$ when $S(c)$ gains a fourth contact pair, so $c$ reaches
the other vertex $v'\in V(\phi)$ incident to $e$.
By Lemma~\ref{lem:edge_vertex},
we have $S(\hat{v'}) \subset S=S(\hat{v})$,
so the new contact pair does not involve points other than $p$ and $q$.
Hence, there are two possibilities:
$\ct_{v'} = \{(p,\mb),(q,\mb),(r,\ml),(r,\mt)\}$
or
$\ct_{v'} = \{(p,\mb),(q,\mb),(r,\ml),(q,\mr)\}$.
In either case, $S(\hat{v'})$ is a stapled $(4,3)$-square,
and thus the lemma follows.
\end{proof}
From Lemma~\ref{lem:vertextype}, we know that each vertex $v'\in V(\phi)$
of stapled $(4,3)$-type has at most one growing edge whose growing direction
is outwards from $v'$.
Hence, for each vertex of stapled $(4,3)$-type,
there can be at most one adjacent vertex whose corresponding square
is larger and has four contact points.
This, together with Lemma~\ref{lem:44-43}, implies that
the number of stapled $(4,4)$-squares
is at most the number of stapled $(4,3)$-squares, which is $O(n^2)$
by Lemma~\ref{lem:43bc-42}.

Lastly, we bound the number of non-stapled $(4,4)$-squares.
Recall that we have so far proved that the number of $4$-squares
whose type is not non-stapled $(4,4)$-type is $O(n^2)$.
\begin{lemma}\label{lem:44-other}
 In the graph $\GSq$ defined in Section~\ref{sec:sq_VD},
 any non-stapled $(4,4)$-square is adjacent to at least one
 $4$-square that is not of non-stapled $(4,4)$-type.
\end{lemma}
\begin{proof}
Recall that the graph $\GSq$ is defined as follows:
each vertex is a $4$-square in $\Sq_4$
and each edge corresponds to a valid interval for some contact type
between two $4$-squares obtained at its endpoints.

Let $S\in \Sq_4(\phi)$ be any non-stapled $(4,4)$-square
with contact type $\ct$.
Let $p_\mt,p_\mr,p_\mb,p_\ml \in P$ be the contact points of $S$
such that $\ct = \{(p_\mt, \mt), (p_\mr, \mr), (p_\mb,\mb), (p_\ml,\ml)\}$.\
Let $\ct_\ma := \ct \setminus \{(p_\ma,\ma)\}$ for
each $\ma \in \{\mt,\mr,\mb,\ml\}$.
Since $\ct$ is of non-stapled $(4,4)$-type,
$\ct_\ma$ has three pinned sides and either $\phi-\epsilon$ or $\phi+\epsilon$
is valid for $\ct_\ma$ for any arbitrarily small $\epsilon>0$
by Lemma~\ref{lem:4sq-3sq}.

Without loss of generality,
assume that $\phi+\epsilon$ is valid for $\ct_\ml$.
Then, $\phi+\epsilon$ is also valid for $\ct_\mr$,
while $\phi-\epsilon$ is valid for $\ct_\mt$ and $\ct_\mb$.
Thus, the degree of $S$ is four in graph $\GSq$.
(See the non-stapled $(4,4)$-type of \figurename~\ref{fig:transition}.)
As we increase $\theta$ from $\theta = \phi$,
we have two regular vertices in $V(\theta)$ corresponding to
squares $S_{\ct_\ml}(\theta)$ and $S_{\ct_\mr}(\theta)$,
and a sliding edge $e\in E(\theta)$ between them,
in the Voronoi diagram $\VD(\theta)$.

Let $I_\ml$ and $I_\mr$ be the valid intervals for
$\ct_\ml$ and $\ct_\mr$ such that $\phi$ is
one endpoint of both $I_\ml$ and $I_\mr$.
Note that $\phi+\epsilon$ is contained in both $I_\ml$ and $I_\mr$.
Let $\phi_\ml$ and $\phi_\mr$ be the endpoints of $I_\ml$ and $I_\mr$,
respectively, other than $\phi$.
Without loss of generality, assume that $\phi< \phi_\ml <\phi_\mr$.
Then, $S_{\ct_\ml}(\phi_\ml)$ is a $4$-square.

We claim that $S_{\ct_\ml}(\phi_\ml)$ is not of non-stapled $(4,4)$-type.
In order to show the claim,
consider the rectangle $R(\theta):=S_{\ct_\ml}(\theta) \cup S_{\ct_\mr}(\theta)$
for any $\phi < \theta < \phi_\ml$.
Observe that the boundary of $R(\theta)$ contains
the four points $p_\mt,p_\mr,p_\mb,p_\ml$ in its each side
and $R(\theta)$ is empty.
This implies that the left side $\ml(S_{\ct_\ml}(\theta))$ of
$S_{\ct_\ml}(\theta)$ cannot be pinned by any point in $P$
for $\phi < \theta < \phi_\ml$.
By continuity and a limit argument, therefore,
we cannot have a contact point on the relative interior of
the left side of $S_{\ct_\ml}(\phi_\ml)$,
so $S_{\ct_\ml}(\phi_\ml)$ is not a non-stapled $(4,4)$-square.
Thus, the lemma is shown.
\end{proof}
Lemma~\ref{lem:4sq-3sq} states that the degree of
each $4$-square in $\Sq_4$ is of degree at most four,
and hence is adjacent to at most four non-stapled $(4,4)$-squares.
Since the number of $4$-squares that is not of non-stapled $(4,4)$-type
is $O(n^2)$ and each non-stapled $(4,4)$-square is adjacent to
one of them by Lemma~\ref{lem:44-other},
we conclude that the number of non-stapled $(4,4)$-squares
is also $O(n^2)$.
This proves
the upper bound of Theorem~\ref{thm:44sq}.

Since $\Sigma_4(n) \leq \sum_{1\leq k\leq 4}\Sigma_{4,k}(n)$,
we have that $\Sigma_4(n) = O(n^2)$.
By Lemma~\ref{lem:4sq_upper_example}, we have $\Sigma_4(n) = \Omega(n^2)$,
completing our proof for the claimed upper bound of Theorem~\ref{thm:4sq}.

\subsection{Lower bounds}
We then turn into proving the lower bounds.
As discussed above, we already have $\sigma_4(n) = \Omega(n)$,
which matches the claimed lower bound in Theorem~\ref{thm:4sq}.
Here, we show $\Omega(n)$ lower bounds
for $\sigma_{4,2}(n)$ and $\sigma_{4,3}(n)$,
and then construct a point set with $O(n)$ $4$-squares.

\begin{lemma} \label{lem:43sq_lower}
 For any integer $n\geq 3$,
 $\sigma_{4,2}(n) = \Omega(n)$ and $\sigma_{4,3}(n) = \Omega(n)$.
\end{lemma}
\begin{proof}
For each point $p\in P$, consider the closest point $q$ from $p$
in the Euclidean distance.
So, the disk centered at $p$ with $q$ on its boundary is empty.
Consider the square $S_{pq}$ one of whose diagonal is segment $pq$,
which is empty and is thus a $(4,2)$-square.
This proves that $\sigma_{4,2}(n) = \Omega(n)$.

Without loss of generality, we assume that
the orientation of the square $S_{pq}$ obtained above is $0\in \OS$,
$p$ lies on the bottom-left corner of $S_{pq}$, and $q$ lies on
the top-right corner of $S_{pq}$,
so the contact type of $S_{pq}$ is $\{(p,\mb),(p,\ml),(q,\mt),(q,\mb)\}$.
For $0 \leq \theta \leq \pi/4$,
let $S(\theta)$ be the square in orientation $\theta$
with at least three contact pairs
$(p,\mb)$,$(p,\ml)$, and $(q,\mr)$.
As $\theta$ continuously increase from $0$ to $\pi/4$,
consider the motion of $S(\theta)$.
If one of the sides of $S(\theta)$ hits a third point in $P$ at $\theta=\phi$,
then $S(\phi)$ is a $(4,3)$-square, so we are done.

Suppose this is not the case, so we reach $\theta = \pi/4$.
Then, $S(\pi/4)$ is a stapled $(4,2)$-square with
contact type $\ct=\{(p,\mb),(p,\ml),(q,\mr), (q,\mb)\}$.
Let $v \in V(\pi/4)$ be the vertex corresponding to $S(\pi/4)$.
Now, we grow $S(\pi/4)$ by moving its center along
a growing edge $e$ incident to $v$.
There are two growing edges that are directed outwards from $v$.
See \figurename~\ref{fig:vertextype}.
Unless $e$ is unbounded, the other vertex $v'$ incident to $e$
corresponds to a $(4,3)$-square.

If both growing edges are unbounded, then there is no point in $P$
above the line through $p$ and $q$.
In this case, we repeat the above process in the opposite direction:
Redefine $S(\theta)$ to be the square with at least three contact pairs
$(p,\mb)$,$(p,\ml)$, and $(q,\mt)$ for $0 \geq \theta \geq -\pi/4$,
consider the motion of $S(\theta)$ by
decreasing $\theta$ from $0$ to $-\pi/4$ 
until we have $S(-\pi/4)$,
and then grow $S(-\pi/4)$ as done above.
If this process again fails to find a third contact point and a $(4,3)$-square,
then we have $P$ consists of only two points $p$ and $q$.
(It is obvious that $\sigma_{4,3}(2) = \Sigma_{4,3}(2) = 0$.)
Hence, if $n\geq 3$, there must be at least one $(4,3)$-square such that
$p$ and $q$ are its contact points
for any $p \in P$ and its closest neighbor $q \in P$,
and thus we have $\sigma_{4,3}(n) = \Omega(n)$ for any $n \geq 3$.
\end{proof}
Lemma~\ref{lem:43sq_lower} implies that
there are always $\Omega(n)$ nontrivial $4$-squares among $n$ points.
%
%

\begin{figure}[tbh]
\begin{center}
\includegraphics[width=.6\textwidth]{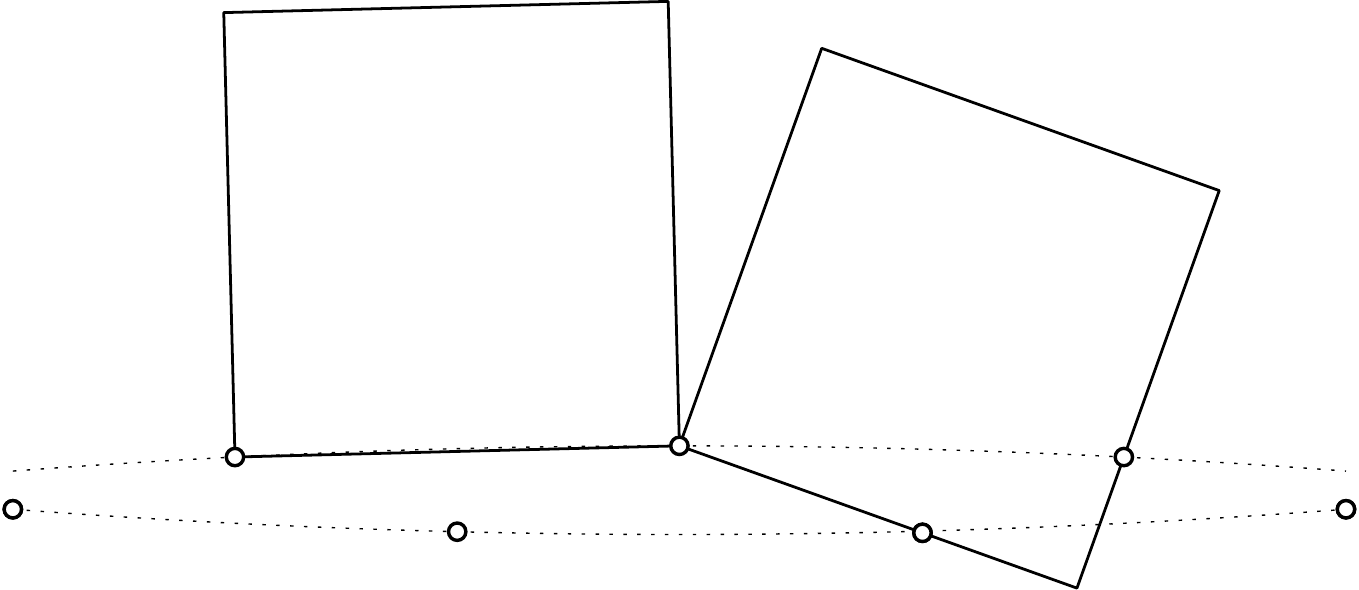}
\end{center}
\caption{Illustration of a point set $P'_n$ with $O(n)$ many $4$-squares.}
\label{fig:4sq_lower_example}
\end{figure}

We finally construct a point set having a small number of $4$-squares.
\begin{lemma} \label{lem:4sq_lower_example}
 For any integer $n \geq 1$,
 there exists a set $P'_n$ of $n$ points such that
 $s_{4,2}(P'_n) = O(n)$, $s_{4,3}(P'_n) = O(n)$, $s_{4,4}(P'_n) = 0$,
 and thus $s_4(P'_n) = O(n)$.
\end{lemma}
\begin{figure}[tbh]
	\begin{center}
		\includegraphics[width=.6\textwidth]{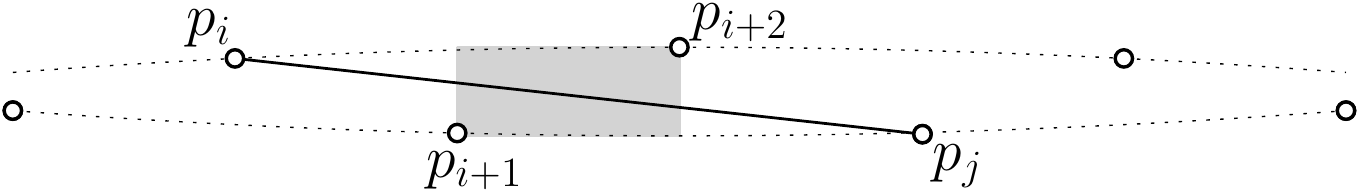}
	\end{center}
	\caption{Proof of Lemma~\ref{lem:4sq_lower_example}}
	\label{fig:4sq_lower_example_proof}
\end{figure}
\begin{proof}
We start with a construction of such a point set $P'_n$ for any $n\geq 1$.
Let $\epsilon\in \Real$ be a sufficiently small positive real number
and $L \in\Real$ be a sufficiently large positive real number.
Consider a big circle $C$ of radius at least $(n+1)L/\epsilon$
and a circular arc $A$ of $C$ with central angle $\epsilon$.
Apply a rigid transformation to $A$, resulting in $A^+$ such that
one endpoint of $A^+$ is located at the origin $o\in\Plane$,
the other lies on the $x$-axis with a positive $x$-coordinate, 
and no point of $A^+$ lies below the $x$-axis.
Let $A^-$ be the mirror image of $A^+$ with respect to the $x$-axis.
Note that $A^+\cup A^-$ lies in between two horizontal lines 
$\{y=-\epsilon/2\}$ and $\{y=\epsilon/2\}$
and the $x$-coordinate of the shared endpoint of $A^+$ and $A^-$ 
other than the origin is at least $nL$ by our construction and simple trigonometry.

For $i=\{1,\ldots, n\}$, let $x_i := iL$.
Let $p_i$ be the point on $A^+ \cup A^-$ such that
$p_i$ is the point with $x$-coordinate $x_i$ on $A^+$ if $i$ is odd,
or on $A^-$ if $i$ is even.
We then define $P'_n := \{p_1, \ldots, p_n\}$.
See Figure~\ref{fig:4sq_lower_example}.

We claim that for any $1\leq i\leq n$ and $i + 2 < j \leq n$,
there is no $4$-square such that both $p_i$ and $p_j$ are its contact points.
This claim immediately implies that $s_{4,4}(P'_n) = 0$
and $s_{4,k}(P'_n) = O(n)$ for $k=2,3$.
In the following, we prove our claim.

Suppose to the contrary that there is a $4$-square $S$
such that both $p_i$ and $p_j$ are its contact points.
Let $\phi\in\OS$ be the orientation of $S$.
Note that the side length of $S$ is at least $3L/\sqrt{2}$ 
since $j-i \geq 3$ and segment $p_ip_j$ is contained in $S$.
Consider the axis-parallel rectangles 
$R^- = [x_{i+1},x_{i+2}]\times [-\epsilon/2, 0]$ 
and $R^+ = [x_{i+1},x_{i+2}]\times [0, \epsilon/2]$.
Note that each of $R^-$ and $R^+$ contains at least one point in $P'_n$
by our construction since $j \geq i+3$.
Since $S$ is empty, a side $\ma(S)$ of $S$ must intersect $R^-$ and
another side $\ma'(S)$ of $S$ must intersect $R^+$, for some $\ma,\ma'\in\{\mt,\mr,\mb,\ml\}$.
We observe that $\ma(S)$ and $\ma'(S)$ are two parallel sides of $S$,
since, otherwise, the boundary of $S$ would not contain $p_i$ or $p_j$.
Hence, the distance between two sides $\ma(S)$ and $\ma'(S)$ is
less than $\sqrt{L^2 + \epsilon^2} < L+\epsilon$,
which is strictly smaller than $3L/\sqrt{2}$, 
the lower bound of the side length of $S$,
leading to a contradiction.
\end{proof}

Consequently, by Lemma~\ref{lem:4sq_lower_example}, we have
$\sigma_{4,2}(n) = \Theta(n)$, $\sigma_{4,3}(n) = \Theta(n)$, and
thus $\sigma_4(n) = \Theta(n)$,
while $\sigma_{4,4}(n) = 0$.
This proves the claimed lower bounds in Theorems~\ref{thm:4sq} and~\ref{thm:44sq}.

\section{Maintaining the $L_\infty$ Voronoi Diagram under Rotation}
\label{sec:alg}

In this section, we present an algorithm
that maintains the combinatorial structure $\VG(\theta)$
of the Voronoi diagram $\VD(\theta)$
while $\theta\in\OS$ continuously increases from $0$ to $\pi/2$.
As observed in Section~\ref{sec:sq_VD},
any combinatorial change of $\VD(\theta)$ corresponds to
a $4$-square, so our algorithm indeed finds
all $4$-squares among the points in $P$.

Lemma~\ref{lem:valid_interval} implies that
for any two consecutive degenerate orientations
$\phi_1, \phi_2 \in \OS$
the vertex set $V(\theta)$ stays the same for all $\phi_1 < \theta < \phi_2$.
By Lemma~\ref{lem:edge_vertex}, a change in the edge set $E(\theta)$
happens when and only when its incident vertices change,
so $\VD(\theta)$ for all $\phi_1 < \theta < \phi_2$ are
combinatorially equivalent.
On the other hand,
if $\phi \in \OS$ is a degenerate orientation,
then $\VD(\phi)$ is not equivalent to $\VD(\theta)$ for
any $\theta\in \OS$ with $\theta\neq\phi$,
since $V(\phi)$ has a non-regular vertex
corresponding to a $4$-square in $\Sq_4(\phi)$
and $V(\theta)$ does not.
Hence, the set $\OS$ of all orientations is divided into
at most $2s_4$ equivalence classes:
open intervals $(\phi_1, \phi_2)$ and singletons $\{\phi\}$
such that $\phi_1$ and $\phi_2$ are any two consecutive degenerate orientations
and $\phi$ is any degenerate orientation.
Summarizing, the combinatorial change of $\VD(\theta)$ happens
at every degenerate orientation $\theta = \phi$ only.

In the following, $\epsilon\in\Real$ denotes any arbitrarily small positive real.
For any degenerate orientation $\phi\in\OS$
and a $4$-square $S\in \Sq_4(\phi)$ with contact type $\ct$,
let $V^-_S \subseteq V(\phi-\epsilon)$ and
$V^+_S \subseteq V(\phi+\epsilon)$ be the sets of regular vertices $v$
in $V(\phi-\epsilon)$ and $V(\phi+\epsilon)$, respectively,
such that $\ct_v \subset \ct$.
For each degenerate orientation $\phi \in \OS$, define
 \[ V^-(\phi) := V(\phi-\epsilon)\setminus V(\phi+\epsilon)
   \quad \text{and} \quad
    V^+(\phi) := V(\phi+\epsilon)\setminus V(\phi-\epsilon)\]
to be the sets of vertices to be deleted and inserted, respectively,
as $\theta$ goes through $\phi$.
Similarly, define
 \[ E^-(\phi) := E(\phi-\epsilon)\setminus E(\phi+\epsilon)
    \quad \text{and} \quad
    E^+(\phi) := E(\phi+\epsilon)\setminus E(\phi-\epsilon).\]
Lemmas~\ref{lem:valid_interval} and~\ref{lem:4sq-3sq} imply the following.
\begin{lemma} \label{lem:change}
 For any degenerate orientation $\phi\in\OS$, the following hold:
 \begin{enumerate}[(i)] \denseitems
  \item $V^-(\phi) = \bigcup_{S\in\Sq_4(\phi)} V^-_S$
    and $V^+(\phi) = \bigcup_{S\in\Sq_4(\phi)} V^+_S$.
  \item $E^-(\phi)$ consists of edges
  incident to a vertex in $V^-(\phi)$ and
  $E^+(\phi)$ consists of edges incident to a vertex in $V^+(\phi)$.
  \item $|V^-(\phi)| + |V^+(\phi)| + |E^-(\phi)| + |E^+(\phi)|
    = \Theta(|\Sq_4(\phi)|)$.
 \end{enumerate}
\end{lemma}
\begin{proof}
As discussed above, Lemma~\ref{lem:valid_interval} directly implies that
$V^-(\phi) = \bigcup_{S\in\Sq_4(\phi)} V^-_S$
and $V^+(\phi) = \bigcup_{S\in\Sq_4(\phi)} V^+_S$.

Any edge $e\in E(\phi-\epsilon)$ still appears in $E(\phi)$ and $E(\phi+\epsilon)$
if and only if the two vertices $u, v\in V(\phi-\epsilon)$ incident to $e$
still appear in $V(\phi)$ and $V(\phi+\epsilon)$.
Hence, $e\in E^-(\phi)$ if and only if $e$ is incident to a vertex
in $V^-(\phi)$.
Analogously, $e\in E^+(\phi)$ if and only if $e$ is incident to a vertex
in $V^+(\phi)$.

Finally,
Lemma~\ref{lem:4sq-3sq} states that $|V^-_S| + |V^+_S|$ is either two or four
for any $4$-square $S\in\Sq_4(\phi)$.
Thus, the number of vertices in $V^-(\phi)$ and $V^+(\phi)$
is at most $4\cdot |\Sq_4(\phi)|$.
Since $E^-(\phi)$ and $E^+(\phi)$ consists of edges that are incident to
these vertices,
we have $|E^-(\phi)|+|E^+(\phi)|$ is at most $20\cdot |\Sq_4(\phi)|$
by Lemma~\ref{lem:vertextype}.
Therefore, we have
$|V^-(\phi)| + |V^+(\phi)| + |E^-(\phi)| + |E^+(\phi)|
    = \Theta(|\Sq_4(\phi)|)$.
\end{proof}

%

\subsection{Events}
Thus, every combinatorial change of $\VD(\theta)$
can be specified by finding all degenerate orientations and
all $4$-squares.
For the purpose, our algorithm handles \emph{events},
defined as follows:
\begin{itemize} \denseitems
 \item An \emph{edge event} is a pair $(e, \phi)$
 for a bounded edge $e\in E(\phi - \epsilon)$ and an orientation $\phi\in\OS$
 such that the embedding $\hat{e} \in \hat{E}(\phi -\epsilon)$ of $e$
 is about to collapse into a point in orientation $\phi$.
 As a result, the two vertices $u, v$ incident to $e$
 are merged into one with contact type $\ct_u \cup \ct_v$,
 and there is a unique $4$-square $S\in\Sq_4(\phi)$
 such that $S = S_{\ct_u}(\phi) = S_{\ct_v}(\phi)$ and its contact type
 is $\ct_u \cup \ct_v$.
 We call $S$ the \emph{relevant} square to this edge event.
 \item An \emph{align event} is a triple $(p, q, \phi)$
 for two distinct points $p, q \in P$ and an orientation $\phi \in \OS$
 such that
 the orientation of the line segment $pq$ is either $\phi$ or $\phi + \pi/2$
 and there is a $4$-square in $\Sq_4(\phi)$ one of whose sides
 contains both $p$ and $q$.
 Each such $4$-square in $\Sq_4(\phi)$ that one of its sides contains
 both points $p$ and $q$ is called \emph{relevant} to this align event.
\end{itemize}
We say that an event \emph{occurs} at $\phi \in \OS$
if its associated orientation is $\phi$.
%

We then observe the following lemmas.
\begin{lemma} \label{lem:event_4sq}
 For any degenerate orientation $\phi\in\OS$, an event occurs at $\phi$.
 More precisely,
 every stapled $4$-square in $\Sq_4(\phi)$ is relevant to an align event
 at $\phi$ and
 every non-stapled $4$-square in $\Sq_4(\phi)$
 is relevant to an edge event at $\phi$.
\end{lemma}
\begin{proof}
Let $\phi \in \OS$ be a degenerate orientation and
$S \in \Sq_4(\phi)$ be a $4$-square in orientation $\phi$.
By the contact type $\ct$ of $S$,
$S$ is either stapled or not.

We first suppose that $S$ is stapled,
so there are two points on a common side of $S$.
Assume without loss of generality that the bottom side $\mb(S)$ of $S$ is stapled,
and $p, q \in P$ are the two contact points on $\mb(S)$,
that is, $(p, \mb), (q, \mb) \in \ct$.
Observe that the segment $pq$ is in orientation $\phi$.
Hence, we have an align event $(p,q,\phi)$ that occurs at $\phi$.

Next, suppose that $S$ is non-stapled.
Since $|\ct| = 4$, all the four sides of $S$ are pinned.
We distinguish three cases depending on the number of contact points:
$\ct$ is either $(4,2)$-type, $(4,3)$-type, or $(4,4)$-type.
In this case, we always have an edge event as shown in
the first three rows of \figurename~\ref{fig:transition}.
More precisely,
as observed in Lemma~\ref{lem:4sq-3sq},
all four different contact types $\ct'\subset\ct$ with $|\ct'|=3$
determine two regular vertices in $V(\phi-\epsilon)$ and
the other two in $V(\phi+\epsilon)$,
for a sufficiently small $\epsilon>0$,
since $S$ is non-stapled.
Let $u,v \in V(\phi-\epsilon)$ be the first two regular vertices
with $\ct_u, \ct_v \subset \ct$.
As shown in the proof of Lemma~\ref{lem:4sq-3sq} and \figurename~\ref{fig:transition},
there is a regular edge $e \in E(\phi-\epsilon)$ between $u$ and $v$
such that $\ct_e = \ct_u \cap \ct_v$.
The squares $S_{\ct_u}(\theta)$ and $S_{\ct_v}(\theta)$ corresponding to
$u$ and $v$ converge to $S_{\ct_u}(\phi)$ and $S_{\ct_v}(\phi)$,
respectively, as $\epsilon$ tends to zero,
and we know that $S_{\ct_u}(\phi) = S_{\ct_v}(\phi) = S$.
Hence, there is an edge event $(e, \phi)$.
\end{proof}

\begin{figure}[tbh]
\begin{center}
\includegraphics[width=.99\textwidth]{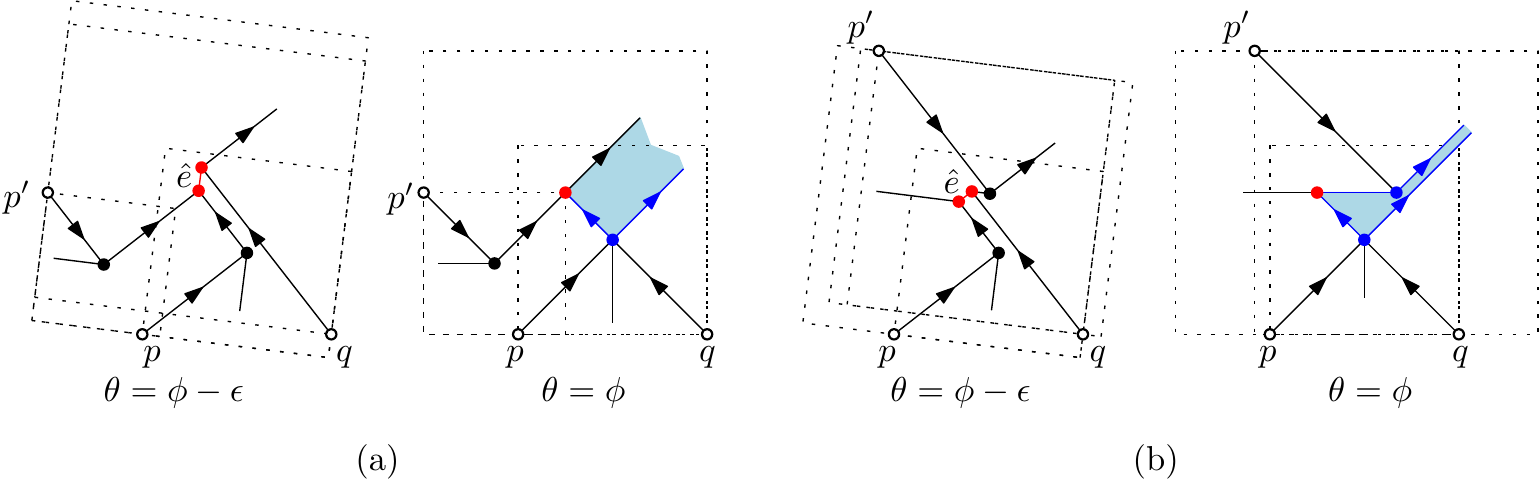}
\end{center}
\caption{Illustration to the proof of Lemma~\ref{lem:inner_alignevent}.
 This figure shows the combinatorial changes around
 the vertex corresponding to a stapled $(4,2)$-square
 and incident edges when
(a) $p'$ lies on the left side of $S'$ and (b) when $p'$ lies on
the top side of $S'$,
at $\theta = \phi-\epsilon$ and $\theta = \phi$.
 The red edge $\hat{e}$ in $\phi-\epsilon$ is collapsed into
 the red vertex in $\phi$.
 }
\label{fig:42align}
\end{figure}

We call an align event $(p,q,\phi)$ \emph{an outer align event}
if both $p$ and $q$ appear consecutively on the boundary of $\OCH(\phi)$
or, otherwise, \emph{an inner align event}.
By Lemma~\ref{lem:unbounded_och},
an outer align event is closely related to the orthogonal convex hull.
If $(p,q,\phi)$ is an outer align event,
then one of $p$ and $q$ was not a vertex of $\OCH(\phi-\epsilon)$
and is about to become a vertex of $\OCH(\phi+\epsilon)$,
or reversely one of the two was a vertex of $\OCH(\phi-\epsilon)$
and is about to disappear from the boundary of $\OCH(\phi+\epsilon)$,
for any arbitrarily small positive $\epsilon>0$.
This implies that
any outer align event corresponds to a combinatorial change of $\OCH(\theta)$
as $\theta$ increases.
Thus, we can precompute all outer align events by using
any existing method that maintains $\OCH(\theta)$ for all $\theta\in\OS$,
such as Alegr\'{i}­a-Galicia et al.~\cite{aosu-obhpps-18}.
\begin{lemma} \label{lem:outer_alignevent}
 There are $O(n)$ outer align events and we can compute
 all outer align events in $O(n\log n)$ time.
\end{lemma}
\begin{proof}
Every combinatorial change of $\OCH(\theta)$ happens when two points $p,q\in P$
appear consecutively on one staircase of $\OCH(\theta)$ and
the orientation of segment $pq$ is either $\theta$ or $\theta+\pi/2$.
Bae et al.~\cite{blacc-cmarchls-09} proved that there are
$O(n)$ combinatorial changes of $\OCH(\theta)$ as $\theta \in \OS$ continuously increases,
and showed how to find them in $O(n^2)$ time.
Later, the time complexity was improved to $O(n\log n)$ time
by Alegr\'{i}­a-Galicia et al.~\cite{aosu-obhpps-18}.

The above discussion shows that
each outer align event comes along with a combinatorial change of
the orthogonal convex hull $\OCH(\theta)$.
Conversely, if there is a combinatorial change of $\OCH(\theta)$
at $\theta = \phi$,
then there exist two points $p,q\in P$ lying on its boundary
such that the orientation of segment $pq$ is
either $\phi$ or $\phi+\pi/2$.
Now, observe that there is a stapled $(4,2)$-square $S$
such that both $p$ and $q$ lie on two corners of $S$
and the interior of $S$ does not intersect with $\OCH(\phi)$.
Hence, we have an outer align event $(p,q,\phi)$.
This implies a one-to-one correspondence
between outer align events and combinatorial changes of $\OCH(\theta)$.
Consequently, there are $O(n)$ outer align events and
we can compute all of them in $O(n\log n)$ time
by the algorithm of Alegr\'{i}­a-Galicia et al.~\cite{aosu-obhpps-18}.
\end{proof}

We then observe the following for inner align events.
\begin{lemma} \label{lem:inner_alignevent}
  If an inner align event occurs at $\phi\in\OS$, then
  there is a stapled $(4,3)$-square in $\Sq_4(\phi)$
  that is relevant to an edge event that occurs at $\phi$.
\end{lemma}
\begin{proof}
We first show that if any align event occurs at $\phi$,
there exists a stapled $(4,2)$-square in $\Sq_4(\phi)$.
For any align event $(p, q, \phi)$,
there exists a stapled $4$-square $S' \in \Sq_4(\phi)$
whose contact type $\ct$ contains pairs $(p, \ma)$ and $(q, \ma)$
for some $\ma \in \{\mt,\mr,\mb,\ml\}$.
Without loss of generality, we assume that $\ma = \mb$
and $p$ is to the left of $q$.
Whichever type of $S'$ is, there is a stapled $(4,2)$-square $S\in \Sq_4(\phi)$
with contact type $\{(p, \mb), (p, \ml), (q, \mb), (q,\mr)\}$.
Indeed, one can check this fact from the five stapled vertex types
other than stapled $(4,2)$-type
listed in Lemma~\ref{lem:vertextype} and \figurename~\ref{fig:vertextype}.

Suppose that $(p, q, \phi)$ is an inner align event,
so $p$ and $q$ are not consecutive along the boundary of $\OCH(\phi)$.
Note that both may appear on two distinct staircases of $\OCH(\phi)$.
By above discussion, there is a stapled $(4,2)$-square $S \in \Sq_4(\phi)$
whose contact points are $p$ and $q$.
Let $v\in V(\phi)$ be the vertex corresponding to $S$.
Without loss of generality, we assume that
the contact type of $S$ is $\ct_v = \{(p,\mb),(p,\ml),(q,\mb),(q,\mr)\}$,
so $p$ lies on the bottom-left corner of $S$ and 
$q$ lies on the bottom-right corner.
Let $S(c)$ for $c\in \Plane$ be the maximal empty square centered at $c$.
We then grow $S=S(\hat{v})$ with its bottom-right corner fixed
by moving $c$ from $\hat{v}$ along the growing non-regular edge $\hat{e}$
directed in the upper-left direction, see the stapled $(4,2)$-type in
\figurename~\ref{fig:vertextype}.
Note that the contact type of $e$ is $\ct_e = \{(p,\mb),(q,\mb),(q,\mr)\}$.
Observe that $e$ is bounded, so we will reach a vertex $v' \in V(\phi)$
incident to $e$ other than $v$.
Otherwise, if $e$ is unbounded, then it defines an empty quadrant
and thus $p$ and $q$ must be consecutive on a staircase of $\OCH(\phi)$
by Lemma~\ref{lem:unbounded_och}, leading to a contradiction.
At $c=v'$, the square $S'=S(\hat{v'})$ gains another contact pair $(p',\ma)$
with a third contact point $p'\in P$ for some $\ma \in \{\mt, \ml\}$,
that is, $\ct_{v'} = \{(p,\mb), (q,\mb),(q,\mr), (p',\ma)\}$.
So, $S'$ is a stapled $(4,3)$-square in $\Sq_4(\phi)$.

There are two cases: either $\ma = \mt$ or $\ma = \ml$.
First, suppose that $\ma = \ml$, as illustrated in
the right of \figurename~\ref{fig:42align}(a) for $\theta=\phi$.
Let $\epsilon>0$ be a sufficiently small positive real,
and set $\theta = \phi-\epsilon$.
Then, there are two vertices $u, u'\in V(\theta)$ such that
$\ct_u = \{(p',\ml), (q,\mr),(q,\mb)\}$
and $\ct_{u'} = \{(p',\ml), (q,\mr),(p,\mb)\}$,
and an edge $e' \in E(\theta)$ between $u$ and $u'$
with $\ct_{e'} = \ct_u \cap\ct_{u'} = \{(p',\ml),(q,\mr)\}$.
Thus, $e'$ is sliding.
See the left of \figurename~\ref{fig:42align}(a).
It is obvious that both squares $S_{\ct_u}(\theta)$ and $S_{\ct_{u'}}(\theta)$
converge to the stapled $(4,3)$-square $S'$ as $\epsilon$ tends to zero,
so $S' = S_{\ct_u}(\theta) = S_{\ct_{u'}}(\theta)$.
Hence, we have an edge event $(e', \phi)$ that occurs at $\phi$
and $S'$ is relevant to it.

Next, we consider the latter case where $\ma = \mt$,
as illustrated in the right of \figurename~\ref{fig:42align}(b).
Then, for $\theta = \phi - \epsilon$,
there are two vertices $u, u'\in V(\theta)$ with
$\ct_u = \{(p',\mt),(q,\mr),(q,\mb)\}$
and $\ct_{u'} = \{(p',\mt),(q,\mr),(p,\mb)\}$,
and an edge $e' \in E(\theta)$ between $u$ and $u'$
with $\ct_{e'} = \ct_u \cap\ct_{u'} = \{(p',\mt),(q,\mr)\}$.
Thus, $e'$ is growing in this case.
See the left of \figurename~\ref{fig:42align}(b).
Observe that both squares $S_{\ct_u}(\theta)$ and $S_{\ct_{u'}}(\theta)$
converge to $S'$ as $\epsilon$ tends to zero.
Hence, we have an edge event $(e, \phi)$ that occurs at $\phi$
and a stapled $(4,3)$-square $S'$ is relevant to it.
\end{proof}

Lemma~\ref{lem:inner_alignevent} implies that
every inner align event can be noticed by handling
an edge event whose relevant $4$-square is of stapled $(4,3)$-type.
This, together with Lemma~\ref{lem:outer_alignevent},
allows us to maintain the diagram $\VD(\theta)$ in an efficient
and output-sensitive way,
as we do not need to test all pairs of points $p,q\in P$
for potential align events.

In order to catch every edge event,
we define the potential edge event as follows:
for any regular orientation $\theta\in\OS$ and
any bounded edge $e = uv \in E(\theta)$,
the \emph{potential edge event} $w(e,\theta)$ for $e$ and $\theta$
is a pair $(e,\phi)$ such that
$\theta < \phi < \pi/2$ and $S_{\ct_u}(\phi) = S_{\ct_v}(\phi)$,
regardless of its emptiness.
If such $\phi$ does not exist, then $w(e,\theta)$ is undefined.
\begin{lemma} \label{lem:potential_event}
 The potential edge event $w(e, \theta)$ is uniquely defined,
 unless undefined.
 Given $e$ and $\theta$, one can decide if $w(e,\theta)$ is defined and
 compute it, if defined, in $O(1)$ time.
\end{lemma}
\begin{proof}
Since the potential edge event $w(e, \theta)$ is
determined only by a constant number of points,
it is obvious that it takes $O(1)$ time to compute it.
In the following, we show how to compute it in details.

Let $u, v \in V(\theta)$ be the vertices incident to $e$.
Note that if one of $u$ and $v$ is a point in $P$,
then the edge $e$ cannot be collapsed into a point,
so $w(e, \theta)$ is undefined in this case.
We thus assume that $u, v \notin P$.
Since $\theta$ is a regular orientation,
$u$ and $v$ are regular and are of either $(3,2)$-type or $(3,3)$-type
by Lemmas~\ref{lem:vertextype}.
Also, $e$ is of one of four different types by Lemma~\ref{lem:regularVD},
as shown in \figurename~\ref{fig:edgetype}.

First, suppose that both $u$ and $v$ are of $(3,2)$-type
or both of them are of $(3,3)$-type.
Note that $e$ is sliding in this case.
We consider the empty rectangle $R(\theta)$ in orientation $\theta$
with contact type $\ct_u \cup \ct_v$.
(Here, we extend our definition of contact type for squares
to that for rectangles.)
Since $e$ is sliding,
$\ct_e = \ct_u \cap \ct_v$ consists of two parallel pinned sides
and the other two sides are pinned by $\ct_u \setminus \ct_e$ and
$\ct_v\setminus \ct_e$,
so the rectangle $R(\theta)$ is uniquely defined.
In principle, we decide if there exists $\phi > \theta$
such that $R(\phi)$ is a square and, if so, compute $\phi$.
This can be done by handling the width and the height of $R(\phi)$
as functions of $\phi$ and by solving $\phi$ when they become equal.
As already known by early research~\cite{cnd-lerps-03,b-cmwsaao-18},
the width and the height functions of $R(\phi)$ are
sinusoidal functions of period $2\pi$
and thus there is at most one possible value for $\phi$
at which $R(\phi)$ is a square.

Second, suppose that $u$ is of $(3,2)$-type and $v$ is of $(3,3)$-type.
Without loss of generality, assume that
$\ct_v = \{(p_1, \mt),(p_2, \mr),(p_3,\mb)\}$ for some $p_1, p_2, p_3 \in P$.
We have two subcases: $e$ is sliding or growing.
\begin{itemize}
 \item If $e$ is sliding, then we have $\ct_e = \{(p_1, \mt), (p_3,\mb)\}$,
 and $\ct_u = \{(p_1, \mt), (p_3,\mb), (p',\mr)\}$ for some $p' \in \{p_1, p_3\}$.
 Assume that $p' = p_1$.
 The other case where $p'=p_3$ can be handled in a symmetric way.
 Then, consider the square $S_{\ct_v}(\theta)$.
 The potential edge event $w(e,\theta)$ is determined by $(e,\phi)$
 such that $p_1$ lies on the top-right corner of $S_{\ct_v}(\phi)$
 only if $\phi > \theta$ and $S_{\ct_v}(\phi)$ is defined.
 Such $\phi$ can be easily found since $\phi+\pi/2$ is
 the orientation of segment $p_1p_2$.
 So, we are done simply by checking
 if there exists the square in orientation $\phi$
 with contact type $\{(p_1, \mt),(p_2, \mr),(p_3,\mb),(p_1,\mr)\}$.
 \item If $e$ is growing, then we have either
 $\ct_e = \{(p_1,\mt),(p_2,\mr)\}$ or $\ct_e = \{(p_2,\mr),(p_3,\mb)\}$.
 Assume the latter case without loss of generality.
 Then, we have either $\ct_u = \{(p_2,\mr),(p_3,\mb),(p_2,\mt)\}$
 or $\ct_u = \{(p_2,\mr),(p_3,\mb),(p_3,\ml)\}$.
 If $\ct_u = \{(p_2,\mr),(p_3,\mb),(p_2,\mt)\}$, then
 the edge $e$ is collapsed at $\phi > \theta$ such that
 $\phi$ is the orientation of segment $p_1p_2$
 if there exists the square in orientation $\phi$ with contact type
 $\{(p_1,\mt),(p_2,\mr),(p_3,\mb),(p_2,\mt)\}$.
 Otherwise, if $\ct_u = \{(p_2,\mr),(p_3,\mb),(p_3,\ml)\}$,
 then consider the rectangle $R(\phi)$ for $\phi > \theta$
 such that $p_3$ lies its bottom-left corner, $p_2$ lies on the right side,
 and $p_1$ lies on the line extending its top side.
 The edge $e$ is collapsed when $R(\phi)$ becomes a square
 with contact type $\{(p_1,\mt), (p_2,\mr),(p_3,\mb),(p_3,\ml)\}$.
 Hence, in either case, we can check whether $e$ is collapsed or not
 and, if so, when $e$ is collapsed.
\end{itemize}
This completes the proof of the lemma.
\end{proof}

%

\subsection{Algorithm}
Our algorithm maintains the combinatorial structure
$\VG(\theta)$ of the Voronoi diagrams $\VD(\theta)$
as $\theta \in \OS$ continuously increases from $0$ to $\pi/2$.
For the purpose, we increase $\theta$ and
stop at every degenerate orientation $\phi$
to find all $4$-squares in $\Sq_4(\phi)$ and
update $\VG(\theta)$ according to the corresponding changes.

For the purpose, we maintain data structures,
keeping the invariants at the current orientation $\theta \in \OS$ as follows.
\begin{itemize}\denseitems
 \item The \emph{graph} $G = (V,E)$ stores the current Voronoi graph
 $\VG(\theta) = (V(\theta),E(\theta))$ into a proper data structure
 that supports insertion and deletion of a vertex and an edge
 in logarithmic time.
 \item The \emph{event queue} $\Q$ is a priority queue that stores
 potential edge events $w(e, \theta)$ for all $e\in E$ and
 all outer align events that occur after $\theta$,
 ordered by their associated orientations.
 Ties can be broken arbitrarily.
 It is implemented by any heap structure, such as the binary heap
 and the binomial heap,
 that supports \emph{find-min}, \emph{insert}, and \emph{delete} operations
 in logarithmic time.
 \item The \emph{search tree} $\T$ is a balanced binary search tree
 on the set $K := P\times \{\mt,\ml,\mb,\mr\}$
 of all contact pairs indexed by any total order on $K$.
 Each node labeled by $(p,\ma)\in K$ stores the set of all regular vertices
 $v\in V$ such that $(p,\ma) \in \ct_v$,
 denoted by $\T(p,\ma)$, into a sorted list $L(p,\ma)$ by
 the order along the boundary of the face of $\VD(\theta)$
 for the contact pair $(p,\ma)$.
%

 By \emph{inserting a vertex} $v$ into $\T$, we mean
 adding $v$ into $L(p, \ma)$ for all $(p, \ma) \in \ct_v$;
 by \emph{deleting a vertex} $v$ from $\T$, we mean
 deleting $v$ from $L(p, \ma)$ for all $(p, \ma) \in \ct_v$.
 We can insert or delete a vertex into or from $\T$ in logarithmic time
 by a binary search on $L(p, \ma)$ after finding the node with label $(p, \ma)$.
\end{itemize}
We also add a sufficient number of pointers between the copies of the same object
so that each of them can be referenced from another across
different data structures in $O(1)$ time.
This makes possible in $O(1)$ time, for examples,
to find a vertex $v$ from $V$ by its copies stored in $\T$ and
to directly access the node in $\Q$ storing an event corresponding to
a given edge $e\in E$.

Note that the structures we maintain stay the same
between any two consecutive degenerate orientations.
In particular, the order of list $L(p, \ma)$ does not change
until an event occurs and hence we reach a degenerate orientation.
Also, note that the space used by the data structures is bounded by $O(n)$
by the invariants.
The search tree $\T$ consists of $4n$ nodes and stores
several copies of each regular vertex in the current vertex set $V(\theta)$.
Since each regular vertex is described by a contact type with three pairs,
a vertex is stored in three different nodes in $\T$.
The number of outer align events is $O(n)$ by Lemma~\ref{lem:outer_alignevent}
and the number of potential edge events stored in $\Q$
is $O(n)$ at any moment by Lemma~\ref{lem:VD_complexity}.

Our algorithm runs in two phases: the initialization and the main loop.
Without loss of generality, we assume that $0\in\OS$ is a regular orientation.
In the initialization phase,
we initialize the data structures for $\theta = 0$.
We first compute $\VG(0)$ by any optimal algorithm computing
the $L_\infty$ Voronoi diagram~\cite{l-tdvdlpm-80,lw-vdl1m2dsa-80}.
Then,
for any regular vertex $v\in V(0)$, we insert $v$ into $V$ and $\T$;
for any bounded edge $e\in E(0)$,
we insert $e$ into $E$,
we compute the potential edge event $w(e, 0)$, if defined,
and insert it into $\Q$.
Compute all outer align events by Lemma~\ref{lem:outer_alignevent}
and insert them into $\Q$.

We are then ready to run the main loop of our algorithm
from the current orientation $\theta = 0$.
In the main loop,
we repeatedly recognize the next degenerate orientation $\phi > \theta$
by finding an event with a smallest associated orientation from $\Q$,
collect all events that occur at $\phi$ by extracting them from $\Q$,
and handle them by performing the following two steps:
(1) computing all $4$-squares in $\Sq_4(\phi)$
and (2) updating our structures properly as $\theta$ proceeds over $\phi$.

\paragraph*{Computing all $4$-squares in $\Sq_4(\phi)$.}
Let $W$ be the set of all events in $\Q$ whose associated orientation
is commonly $\phi$.
The set $W$ can be obtained by repeatedly performing operations
on the event queue $\Q$; check if the associated orientation
of the minimum element in $\Q$ is exactly $\phi$ and extract it, if so.
For each event $w\in W$, we find all squares relevant to $w$,
according to the type of $w$.
We initialize $\Sq_4(w)$ to be an empty set as a variable,
and will finally consist of all $4$-squares relevant to $w$.

First, suppose that $w = (e, \phi)$ is an edge event.
Then, $e$ is an edge in $E = E(\theta)$ for the current orientation
$\theta < \phi$.
Let $u, v\in V=V(\theta)$ be the two vertices incident to $e$.
Note that $V = V(\theta')$ and $E = E(\theta')$ for all $\theta<\theta'<\phi$
by Lemma~\ref{lem:valid_interval}.
There is a unique $4$-square $S$ in $\Sq_4(\phi)$ relevant to
the edge event $w$
such that the contact type of $S$ is $\ct_u \cup \ct_v$.
So, we can specify $S$ in $O(1)$ time and add it to $\Sq_4(w)$.

Note that $S$ is a non-stapled $4$-square or
a stapled $(4,3)$-square (of type $(4,3)$-(b) or $(4,3)$-(c)),
as observed in Lemma~\ref{lem:4sq-3sq} and \figurename~\ref{fig:transition}.
If $S$ is a stapled $(4,3)$-square with $p,q\in P$ on a common side of $S$,
then an inner align event $(p,q,\phi)$ also occurs simultaneously at $\phi$.
In this case, we add the align event $(p, q, \phi)$ into $W$.
Hence, we do not miss any inner align event by Lemma~\ref{lem:inner_alignevent}.

Next, suppose that $w = (p, q, \phi)$ is an align event.
In this case, there can be several $4$-squares in $\Sq_4(\phi)$ relevant to $w$,
and the number indeed can be $\Omega(n)$.
By definition of align events,
the orientation of segment $pq$ is either $\phi$ or $\phi+\pi/2$.
Without loss of generality, we assume that segment $pq$ is in orientation $\phi$
and $p$ is to the left of $q$ in orientation $\phi$,
so any $4$-square relevant to $w$ contains both $p$ and $q$
on its top or bottom side.
Note that all squares relevant to $w$ are stapled.

There are two possibilities for $S$ by Lemma~\ref{lem:4sq-3sq}.
(1) $V^-(S)$ is empty
or (2) $V^-(S)$ consists of a single regular vertex in $v\in V$.
See \figurename~\ref{fig:transition} for an illustration.
We first find all squares of the latter case (2).
In this case, the contact type $\ct_v$ of the unique vertex $v\in V^-(S)$
is a subset of $\ct$, and we indeed have $S_{\ct_v}(\phi) = S$.
All those squares $S$ in case (2) can be found by
the square-hit-by-point query.
In a \emph{square-hit-by-point} query at the current orientation $\theta$,
we are given a contact pair $(p, \ma) \in K$
and a point $a\in\Plane$ such that
the orientation of segment $pa$ is $\phi$ or $\phi+\pi/2$
and there is no event in open interval $(\theta,\phi)$,
the goal is to report all vertices $v \in \T(p,\ma)$ such that
its corresponding square $S_{\ct_v}(\phi)$ contains
both $p$ and $q$ on the specified side $\ma(S_{\ct_v}(\phi))$.

\begin{figure}[tbh]
\begin{center}
\includegraphics[width=.95\textwidth]{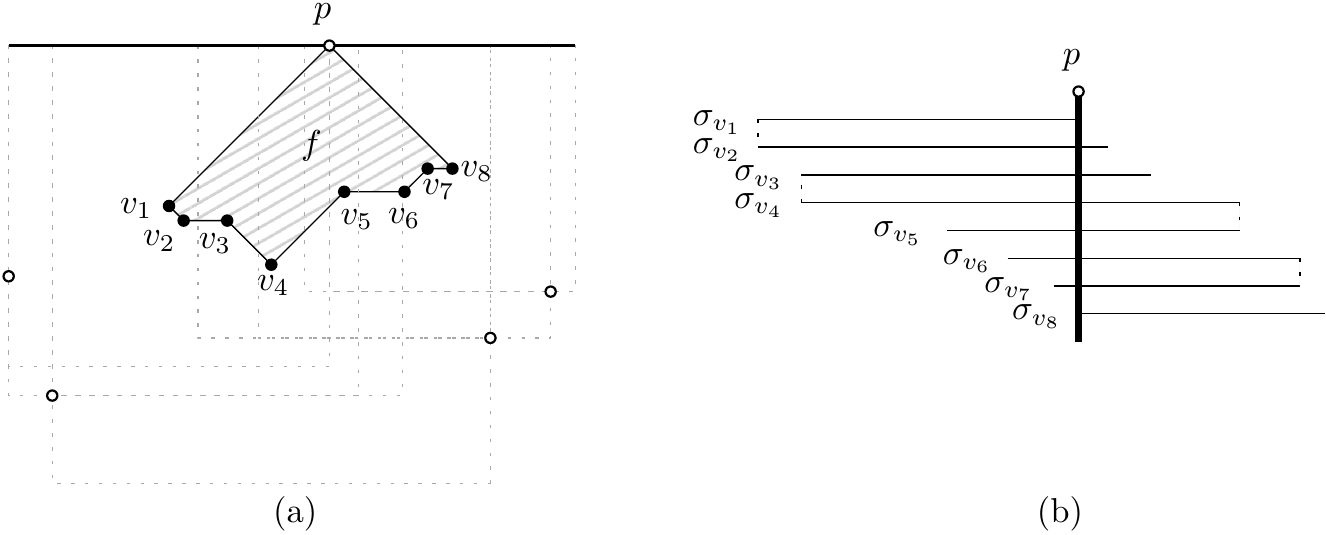}
\end{center}
\caption{(a) A typical example of the face $f$ of $\VD(\phi)$ for $(p,\mt)$
 in orientation $\phi$.
 In the list $L(p,\mt)$, the vertices $v_1, \ldots, v_8$ along the boundary
 of $f$ are stored in that order.
 (b) The top sides $\sigma_{v_1}, \ldots, \sigma_{v_8}$ of
 the squares corresponding to vertices $v_1, \ldots, v_8$ have
 a total order along the line extending the sides and the order is
 the same as that for vertices in the list $L(p, \mt)$.
 }
\label{fig:squarehit_query}
\end{figure}
This type of queries can be efficiently answered using the search tree $\T$
and its invariant.
\begin{lemma} \label{lem:squarehit_query}
 A square-hit-by-point query can be processed in $O(\log n + t)$ time,
 where $t$ denotes the number of reported vertices in $V$.
\end{lemma}
\begin{proof}
For a square-hit-by-point query,
we are given a contact pair $(p, \ma)$ and $q\in\Plane$
such that the orientation of segment $pq$ is $\phi$ or $\phi+\pi/2$
and there is no event in open interval $(\theta,\phi)$.
By our construction,
$\T(p,\ma)$ consists of the vertices on the boundary of the face $f$ for $(p,\ma)$
and
$L(p,\ma)$ stores them in the order along that boundary.
See \figurename~\ref{fig:squarehit_query}(a).
Also, note that this order remains the same up to $\phi$,
since there is no event in $(\theta, \phi)$.
This implies that the squares $S_{\ct_v}(\phi)$ corresponding to
$v \in \T(p,\ma)$ in orientation $\phi$ are well defined.

We then consider the side $\ma(S_{\ct_v}(\phi))$
of the squares $S_{\ct_v}(\phi)$ for each $v\in \T(p,\ma)$,
denoted by $\sigma_v$.
Note that $\sigma_v$ is a line segment that contains $p$.
For any two $v,v' \in \T(p,\ma)$,
two segments $\sigma_v$ and $\sigma_{v'}$ are
not properly nested,
since both squares $S_{\ct_v}(\phi)$ and $S_{\ct_{v'}}(\phi)$ are empty
and have $p$ as a contact point.
That is, if $\sigma_v \subset \sigma_{v'}$, then
exactly one endpoint of $\sigma_v$ is shared by an endpoint of $\sigma_{v'}$.
This implies that
the segments $\sigma_v$ themselves are sorted in the same order of $v$
in $L(p,\ma)$.
See \figurename~\ref{fig:squarehit_query}(b).

Based on the above observations,
we can process the square-hit-by-point query
by a binary search on the sorted list $L(p,\ma)$.
Without loss of generality,
we assume that $\ma = \mt$ and $q$ lies to the right of $p$.
First, find the node of $\T$ labeled by $(p, \mt)$ in $O(\log n)$ time.
Next, we do a binary search on the list $L(p,\mt)$
to find the rightmost vertex $u$ in the list $L(p,\mt)$
such that the top side $\sigma_{u}$ of the corresponding square
$S_{\ct_{u}}(\phi)$ contains $q$.
More precisely, this can be done by computing the coordinate of the right endpoint
of $\sigma_u$ for the median vertex $u$ in $L(p,\mt)$
and testing if $q$ is to the right or to the left of the endpoint.
Then, all the vertices in $L(p,\mt)$ to the right of $u$, including $u$,
have the same property.
So, we report all these vertices from $u$ by traversing $L(p,\mt)$
to its right end.
Hence, we can answer the query in $O(\log n)$ time plus the time
proportional to the number of reported vertices.
\end{proof}


We process four square-hit-by-point queries
for $((p,\mt),q)$, $((p,\mb, q))$, $((q,\mt),p)$, and $((q,\mb),p)$.
Let $U \subset V$ be the set of the reported vertices in the current $V$
by these four queries.
For every $v\in U$, add $S_{\ct_v}(\phi)$ into $\Sq_4(w)$.
In this way, we collect all $4$-squares relevant to $w$ of case (2).

We then find relevant $4$-squares of case (1).
If $S$ is a square of case (1),
then $S$ is of stapled $(4,3)$-(b) or $(4,3)$-(c)-type.
In addition, if the top side of $S$ is stapled by $p$ and $q$,
then the contact type $\ct$ of $S$ includes three contact pairs
$(p,\mt)$, $(q,\mt)$, and $(q,\mr)$ involving $p$ or $q$,
so it is unique, if exists.
The fourth contact pair of $S$ can be found
from the rightmost vertex reported by processing the square-hit-by-point query
for $(p,\mt)$ and $q$.
Analogously, we can find such a square whose bottom side is stapled.
In this way, we find the squares of case (1) and add them into $\Sq_4(w)$.

We have the following.
\begin{lemma} \label{lem:alg_square_finding}
 The first step of the main loop as described above
 takes time $O(|\Sq_4(\phi)| \log n)$ and
 we have $\Sq_4(\phi) = \bigcup_{w \in W} \Sq_4(w)$ at last.
\end{lemma}
\begin{proof}

As argued above,
we correctly compute and add all $4$-squares relevant to each event $w\in W$
into $\Sq_4(w)$.
In order to see that $\Sq_4(\phi) = \bigcup_{w \in W} \Sq_4(w)$,
we show that $W$ is indeed the set of all events that occur at $\phi$.
Since we precomputed all outer align events and inserted them into $\Q$,
it is obvious that $W$ contains all outer align events that occur at $\phi$.

Note that, in the main loop, $\theta$ is the current orientation
and $\phi > \theta$ is the smallest orientation of elements stored in
the current $\Q$.
We claim that any potential edge event whose associated orientation
is $\phi$ is an edge event that indeed occurs at $\phi$.
Let $w = (e,\phi)$ be such a potential edge event.
Every outer align event is correctly specified
and every inner align event occurs simultaneously with an edge event
by Lemma~\ref{lem:inner_alignevent}.
Hence, there is no align event before $\phi$.
Also, since an edge event is also a potential edge event
and $w$ is the earliest one among those potential edge events,
all orientations $\theta'$ such that $\theta < \theta' < \phi$
are regular.
By Lemma~\ref{lem:valid_interval}, this implies that
$\phi-\epsilon$ for any arbitrarily small positive $\epsilon>0$
is valid for both vertices incident to $e$.
Hence, $w$ indeed occurs as an edge event.

By handling each edge event,
we correctly recognize an inner align event by adding it into $W$.
Lemma~\ref{lem:inner_alignevent} guarantees that
all inner events that occur at $\phi$ can be captured
by edge events that occur at $\phi$.
This proves that $W$ is the set of all events that occur at $\phi$
and thus $\Sq_4(\phi) = \bigcup_{w \in W} \Sq_4(w)$.

We then analyze the time complexity.
In the first step of the main loop,
we first perform $O(|W|)$ operations on $\Q$, spending $O(|W| \log n)$ time.
For each $w\in W$, if $w$ is an edge event,
then $\Sq_4(w)$ consists of a single square
and we find it in $O(1)$ time.
If $w$ is an align event,
then we spend $O(\log n + |\Sq_4(w)|)$ time by Lemma~\ref{lem:squarehit_query}.
Finally, observe that $|W| = O(|\Sq_4(\phi)|)$
since any $4$-square in $\Sq_4(\phi)$ is relevant to at most
two different events in $W$.
Therefore,
we spend $O(|\Sq_4(\phi)| \log n)$ time over all $w\in W$.
\end{proof}

\paragraph*{Updating the structures.}
In this step, we first specify
the sets $V^-(\phi)$, $V^+(\phi)$, $E^-(\phi)$, and $E^+(\phi)$,
and then update our data structures to keep the invariants accordingly.
This can be done from the set $\Sq_4(\phi)$
obtained in the previous step by Lemma~\ref{lem:alg_square_finding}.

For each $S \in \Sq_4(\phi)$,
we compute $V^-_S$ and $V^+_S$ by Lemma~\ref{lem:4sq-3sq}
and its proof as illustrated in \figurename~\ref{fig:transition}.
By Lemma~\ref{lem:change}(i),
we obtain $V^-(\phi)$ and $V^+(\phi)$.
By Lemma~\ref{lem:change}(ii),
we can compute the edge sets $E^-(\phi)$ and $E^+(\phi)$
by searching the neighbors of $V^-(\phi)$ in $\VG = \VG(\theta)$
together with those vertices in $V^-(\phi)$ and $V^+(\phi)$.
\begin{lemma} \label{lem:alg_modifiededge}
 The sets $E^-(\phi)$ and $E^+(\phi)$ can be found
 in $O(|\Sq_4(\phi)| \log |\Sq_4(\phi)|)$ time.
\end{lemma}
\begin{proof}
Note that the number of vertices in $V^-(\phi) \cup V^+(\phi)$
is $O(|\Sq_4(\phi)|)$ by Lemma~\ref{lem:change}(iii).
The set $E^-(\phi)$ can be easily found from $E$
in time proportional to $|V^-(\phi)|$ by Lemma~\ref{lem:change}
since the degree of any vertex is bounded by a constant
by Lemma~\ref{lem:vertextype}.
Let $N \subset V \setminus V^-(\phi)$ be the set of vertices
adjacent to some $v\in V^-(\phi)$.
Hence, each $v\in N$ is incident to some $e\in E^-(\phi)$.
Note that $|N| = \Theta(|V^-(\phi)|)$ by Lemma~\ref{lem:vertextype}.

To find $E^+(\phi)$, we need a concept of \emph{half-edges},
informally being two halves of an edge when cut into two.
Note that if a vertex $v$ exists in $V(\theta)$, then
all the half-edges incident to $v$ can be specified locally
without knowing its opposite incident vertex
as shown in Lemma~\ref{lem:vertextype} and \figurename~\ref{fig:vertextype}.

For a set $U$ of vertices, let $H(U)$ be the set of half-edges incident to
a vertex $v \in U$.
Then, every edge in $E^-(\phi)$ can be specified by matching half-edges
in $H(V^+(\phi)) \cup H(N)$.
This can be done in $O((|V^-(\phi)|+|V^+(\phi)|) \log (|V^-(\phi)|+|V^+(\phi)|))$
by sorting those half-edges in $H(\Vin^-) \cup H(N)$ by
their contact type with two or three contact pairs.
If there are some unmatched half-edges in $H(V^+(\phi)) \cup H(N)$,
then they are unbounded edges.
Therefore, we can specify $E^-(\phi)$ and $E^+(\phi)$
in the claimed time bound.
\end{proof}

We are ready to update our structures for $\phi+\epsilon$
for any arbitrarily small $\epsilon>0$.
Note that we currently have $V = V(\theta) = V(\phi-\epsilon)$ and
$E = E(\theta) = E(\phi - \epsilon)$.
We update $V$ and $E$ as follows:
delete all vertices in $V^-(\phi)$ from $V$
and all edges in $E^-(\phi)$ from $E$, and
then insert all vertices in $V^+(\phi)$ into $V$
and all edges in $E^+(\phi)$ into $E$.
Then, update $\T$ and $\Q$ as follows:
We delete each $v \in V^-(\phi)$ from $\T$
and insert each $v\in V^+(\phi)$ into $\T$.
For each $e\in E^-(\phi)$, we delete the potential edge event for $e$ from $\Q$.
For each $e\in E^+(\phi)$,
we compute the potential edge event $w(e, \phi+\epsilon)$
by Lemma~\ref{lem:potential_event} and insert it into $\Q$, if defined.
Lastly, set $\theta$ to be $\phi+\epsilon$.

\begin{remark_nonumber}
If one wants to have the diagram $\VD(\phi)$ in the degenerate orientation $\phi$,
then $V(\phi)$ can be obtained by deleting all vertices in $V^-(\phi)$
from $V=V(\phi-\epsilon)$ and inserting vertices corresponding to
$4$-squares in $\Sq_4(\phi)$ into $V$,
while $E(\phi)$ can be obtained by deleting all edges in $E^-(\phi)$
from $E=E(\phi-\epsilon)$
and inserting those incident to any new vertex into $E$,
which can be found by the same method described in Lemma~\ref{lem:alg_modifiededge}.
This dose not increase the total asymptotic time complexity of
our algorithm.
\end{remark_nonumber}

We finally conclude the following.
\begin{theorem} \label{thm:alg_VD}
 Given a set $P$ of $n$ points in general position,
 the total amount of combniatorial changes of
 the $L_\infty$ Voronoi diagram of $P$ while the axes rotates by $\pi/2$
 is bounded by $\Theta(s_4)$,
 where $s_4$ denotes the number of $4$-squares among $P$.
 The combinatorial structure of the Voronoi diagram can be maintained
 explicitly in total $O(s_4 \log n)$ time using $O(n)$ space.
\end{theorem}
\begin{proof}
The bound $\Theta(s_4)$ on the total amount of combinatorial changes
of $\VD(\theta)$ is shown by Lemma~\ref{lem:change}(iii).

The correctness of our algorithm maintaining
the combinatorial structure of $\VD(\theta)$ is already established
by the above discussions,
including Lemmas~\ref{lem:change} and~\ref{lem:alg_square_finding}.

For the time complexity,
note that the initialization phase takes $O(n\log n)$ time.
In the main loop, it takes $O(|\Sq_4(\phi)| \log n)$
for each degenerate orientation $\phi\in\OS$
by Lemmas~\ref{lem:alg_square_finding}
and~\ref{lem:alg_modifiededge}, together with the observation
that the number of operations performed to our structures
is bounded by $O(|\Sq_4(\phi)|)$.
Summing up this over all degenerate orientations $\phi$,
we have $O(s_4 \log n)$,
since $\sum_{\phi} |\Sq_4(\phi)| = |\Sq_4| = s_4$.
\end{proof}

By maintaining the Voronoi diagram $\VD(\theta)$,
it is now obvious that we can indeed compute all $4$-squares
in an output-sensitive way.
\begin{corollary} \label{coro:alg_4sq}
 Given a set $P$ of $n$ points in general position,
 we can compute all $4$-squares among points in $P$
 in $O(s_4 \log n)$ time.
\end{corollary}

\section{Maximal Empty Squares} \label{sec:mes}

An empty square among $P$ is \emph{maximal} if there is no other empty square
that contains it.
The following is a well known fact.
\begin{lemma} \label{lem:MES1}
For any $\theta\in\OS$, a maximal empty square in $\theta$
is centered at a vertex or a point on an edge of $\VD(\theta)$.
\end{lemma}
\begin{proof}
Let $S$ be a maximal empty square in orientation $\theta$.
Suppose to the contrary that the center of $S$ lies in a face of $\VD(\theta)$.
Then, the number of pinned sides of $S$ is exactly one
and its contact type $\ct$ is either $\{(p,\ma)\}$ or $\{(p,\ma),(q,\ma)\}$
for some $p,q\in P$ and $\ma\in\{\mt,\mr,\mb,\ml\}$,
by our general position assumption.
Without loss of generality, assume that $\ma=\mb$.
In either case, we can grow $S$ by slightly moving its center upwards,
keeping its contact type, a contradiction. 
\end{proof}

Indeed, if $S$ is a maximal empty square centered at a point on an edge of $\VD(\theta)$,
then the edge should be sliding by Lemma~\ref{lem:edge_vertex}.
We, however, handle all the edges together, since taking only sliding edge into account
does not reduce the asymptotic complexity of our algorithms.

Let $\mathcal{V}:=\bigcup_{\theta\in\OS} V(\theta)$ and 
$\mathcal{E}:=\bigcup_{\theta\in\OS} E(\theta)$.
For each regular vertex $v\in \mathcal{V}$ and each valid interval $I$ for $\ct_v$,
the pair $(v, I)$ is called an \emph{MES class} for $v$;
for each regular edge $e = uv\in \mathcal{E}$ and each nonempty interval $I$,
the pair $(e, I)$ is called an \emph{MES class} for $e$
if $I$ is the intersection of a valid interval for $\ct_u$ and a valid interval for $\ct_v$.
Let $\MES_V$ and $\MES_E$ be the set of all MES classes for $v\in \mathcal{V}$
and those for $e\in \mathcal{E}$, respectively. 
For each MES class $\Mes = (v, I) \in \MES_V$,
define $\Mes(\theta) := S_{\ct_v}(\theta)$ for $\theta \in I$.
For each MES class $\Mes = (e, I) \in \MES_E$ such that $e$ is between $u,v\in \mathcal{V}$,
define $\Mes(\theta) := S_{\ct_u}(\theta) \cup S_{\ct_v}(\theta)$.
Observe that every maximal empty square is either equal to $\Mes(\theta)$ for some $\Mes\in\MES_V$
or contained in $\Mes(\theta)$ for some $\Mes \in\MES_E$.
Hence, in this way, all maximal empty squares are described by the MES classes.

Lemma~\ref{lem:4sq-3sq} implies 
that $\MES_V$ and $\MES_E$ consist of $O(s_4) = O(n^2)$ MES classes.
By our algorithm described above, we can compute all MES classes in $\mathcal{M}$
in $O(s_4 \log n)$ time.
\begin{lemma} \label{lem:MES}
 Given a set $P$ of $n$ points in the general position,
 the number of MES classes in $\MES_V \cup \MES_E$ is $O(s_4)$ and 
 we can compute all the MES classes in $O(s_4 \log n)$ time.
\end{lemma}
\begin{proof}
The number of MES classes is bounded by $O(s_4)$ by Lemma~\ref{lem:4sq-3sq}
and Theorem~\ref{thm:alg_VD}.
In the algorithm described and summarized in Theorem~\ref{thm:alg_VD},
we can actually collect all the vertices in $\mathcal{V}$ with their valid intervals
and all the edges in $\mathcal{E}$ in the same time bound $O(s_4 \log n)$.
Therefore, we can compute all the MES classes in $\MES_V \cup \MES_E$
in $O(s_4 \log n)$ time. 
\end{proof}

\subsection{Computing a largest empty square}

Here, we discuss the problem of computing a largest empty square over all orientations among $P$.
Since any point $c$ arbitrarily far from all points in $P$ determines an empty square
of arbitrarily large size,
we need to restrict candidate optimal squares not to be unbounded,
so the problem becomes nontrivial.
There can be several such ways, but the problem is in principle an optimization 
over all MES classes under consideration of boundary cases.
We demonstrate two such variants.
\begin{enumerate}[(1)] \denseitems
  \item Find a largest empty square among all those with three or four pinned sides.
  \item  Find a largest empty square among all those in orientation $\theta$ contained in
    the bounding box $B(\theta)$ over all orientations $\theta\in\OS$,
    where $B(\theta)$ denotes the smallest rectangle in orientation $\theta$ that encloses $P$.
\end{enumerate}
For the rectangle variant of this problem,
Chaudhuri et al.~\cite{cnd-lerps-03} restricted an empty rectangle
to have four contact pairs, one on each side,
and Bae~\cite{b-marsap-19X} considered those contained in
the bounding box $B(\theta)$ with an application
to the problem of computing a minimum rectangular or square annulus
in arbitrary orientation.
We show that both variants can be solved in $O(n^2\log n)$ time in the worst case.

For the first variant, we can just search MES classes in $\MES_V$
by Lemma~\ref{lem:edge_vertex}.
For each MES class $\Mes = (v,I) \in\MES_V$, we maximize the radius of $\Mes(\theta)$
over all $I\in \theta$, including its endpoints by checking the limit of $\Mes(\theta)$
at each endpoint of $I$.
Hence, we do not miss any $4$-square by Lemma~\ref{lem:4sq-3sq}.
We can just return the maximum of maximum squares over all $\Mes \in \MES_V$.
This takes $O(s_4)$ time after computing all MES classes in $\MES_V$,
so $O(s_4 \log n)$ time in total.

For the second variant, 
observe that the bounding box $B(\theta)$ is described by its four contact points.
We first compute all changes of the contact points of $B(\theta)$
by computing the convex hull of $P$ and the antipodal pairs of extreme points~\cite{t-sgprc-83}.
This results in a partition of $\OS$ into $h = O(n)$ intervals whose endpoints 
$0 \leq \theta_1 < \theta_2 < \cdots < \theta_h < \pi/2$ 
correspond to a change of the contact pairs of $B(\theta)$.
Define $\theta_0 := \theta_h$ since the orientation space $\OS$ is cyclic.
For each MES class $\Mes = (e, I) \in \MES_E$ for edge $e$, 
let $\Mes_i := (e, I\cap (\theta_{i-1}, \theta_i); i)$ for each $i=1,\ldots h$.
We call each $\Mes_i$ a \emph{boxed MES class} if $I\cap (\theta_{i-1}, \theta_i) \neq \emptyset$.
We collect all boxed MES classes into set $\bMES_E$ in $O(s_4 \log n + n^2)$ time.
\begin{lemma} \label{lem:boxedMES}
 We can compute the set $\bMES_E$ of all boxed MES classes in $O(s_4 \log n + n^2)$ time.
\end{lemma}
\begin{proof}
For each MES class $\Mes = (e, I) \in \MES_E$,
we can do a binary search for the endpoints of $I$ on the sorted list of
$\theta_1, \theta_2, \ldots, \theta_h$.
So, for each $\Mes \in \MES_E$, we can find all $\Mes_i$ that are boxed MES classes
in $O(\log h + t_\Mes) = O(\log n + t_\Mes)$ time,
where $t_\Mes$ denotes the number of those $\Mes_i$ that are boxed MES classes.
By Lemma~\ref{lem:MES},
the total time is hence $O(s_4 \log n + t)$, where $t = \sum_{\Mes\in\MES_E} t_\Mes$.
Finally, observe that $t = O(hn) = O(n^2)$
since the number of valid intervals that are cut at $\theta_i$ is bounded by $O(n)$
by Lemma~\ref{lem:VD_complexity}.
\end{proof}
Finally, for each $\Mes_i = (e, I_i; i) \in \bMES_E$ with $\Mes=(e, I)\in\MES_E$,
let $\Mes_i(\theta) := \Mes(\theta) \cap B(\theta)$ 
for $\theta \in I_i = I\cap (\theta_{i-1},\theta_i)$.
Note that in the range of $\theta$ the four contact points of $B(\theta)$ is fixed,
so $\Mes_i(\theta)$ as a function of $\theta$ is of constant complexity.
So, we can find a largest square contained in $\Mes_i(\theta)$ over all 
$\theta \in I\cap (\theta_{i-1},\theta_i)$ in $O(1)$ time.
Therefore, the second variant of the largest empty square problem can be solved
in $O(s_4\log n + n^2)$ time.

\begin{theorem} \label{thm:MES}
 Given a set $P$ of $n$ points in the plane,
 a largest empty square among $P$ in arbitrary orientation
 can be computed in worst-case $O(n^2 \log n)$ time.
\end{theorem}

Some query versions of the problem can also be considered.
\begin{theorem} \label{thm:MES_query1}
 Given a set $P$ of $n$ points in general position,
 in $O(n^2 \log n)$ time,
 we can preprocess $P$ into a data structure of size $O(n^2 \alpha(n))$
 that answers the following query in $O(\log n)$ time:
 given an orientation $\beta\in\OS$, find
 a largest empty square in orientation $\beta$.
\end{theorem}
\begin{proof}
For any $\theta\in\OS$,
let $\rho(\theta)$ be the radius of a largest empty square
in orientation $\theta$ under consideration of any of the above two variants.
From the above discussions,
it is not difficult to see that
the function $\rho$ is the upper envelope of 
the possible radii over all MES classes.
Each of curve pieces in consideration is sinusoidal with period $2\pi$,
and thus any two of them cross at most once.
This implies that $\rho$ corresponds to a Davenport--Schinzel sequence
of order three.
Hence, $\rho$ is piecewise sinusoidal with $O(n^2\alpha(n))$ breakpoints,
and can be explicitly computed in $O(n^2 \log n)$ time~\cite{h-fuenls-89}.
A query of this type can be simply processed by evaluating 
the function value of $\rho$ in $O(\log n)$ time.
\end{proof}

\begin{theorem}\label{thm:MES_query2}
 Given a set $P$ of $n$ points in the general position,
 in $O(s_4\log n)$ time, 
 we can preprocess $P$ into a data structure of size $O(s_4)$
 that answers the following query in $O(\log n)$ time:
 given a point $c\in\Plane$ and an orientation $\beta\in\OS$,
 find a largest empty square in orientation $\beta$ centered at $c$.
\end{theorem}
\begin{proof}
We first show how to build such a query structure
using $O(s_4 n)$ space.
By Lemma~\ref{lem:4sq-3sq} and Theorem~\ref{thm:alg_VD}, 
we know that there are at most $O(s_4)$ combinatorially distinct
Voronoi diagrams $\VD(\theta)$ over all $\theta \in \OS$.
For each combinatorially distinct diagram $\VD(\theta)$,
we build a point location structure $\mathcal{D}(\theta)$ on $\VD(\theta)$
in such a way that each vertex in $V(\theta)$ and edge in $E(\theta)$
are stored in $\mathcal{D}(\theta)$ as a function of $\theta$.
Then, after a binary search for $\beta$ in the sorted list of 
all degenerate orientations, we can process a point location query
on $\mathcal{D}(\theta)$ in a parametric way.
This requires $O(s_4n)$ space by Lemma~\ref{lem:VD_complexity},
while all the point location structures $\mathcal{D}(\theta)$
can be obtained in $O(s_4 \log n)$ time
since the total amount of changes in $\mathcal{D}(\theta)$
is also bounded by $O(s_4)$ by Theorem~\ref{thm:alg_VD}.

To reduce the storage usage, 
we store $\mathcal{D}(\theta)$ in a persistent way~\cite{dsst-mdsp-89}
as $\theta$ increases from $0$ to $\pi/2$,
along with the execution of our algorithm described in
Theorem~\ref{thm:alg_VD}.
Since the total amount of changes is bounded by $O(s_4)$,
we need only $O(s_4)$ space. 
\end{proof}

\subsection{Computing a minimum square annulus}

A \emph{square annulus} is the closed region between two concentric
squares in a common orientation.
The \emph{minimum square annulus problem in arbitrary orientation} asks to find 
a square annulus that encloses given points $P$ 
and minimizes its width or area over all orientations.
The width of a square annulus is the difference of the radii of its two defining squares.

To solve this problem, Bae~\cite{b-cmwsaao-18,b-marsap-19X} observed that
there exists an optimal annulus such that its outer square is a smallest bounding square $S$ for $P$
in some orientation $\theta \in \OS$ and its inner square is a largest empty square $S'$
whose center restricted to be in the set of possible centers of $S$.

Let $r(\theta)$ be the radius of a smallest bounding square for $P$ in orientation $\theta$ and
$\ell(\theta)$ be the set of centers of all those smallest bounding squares.
Then, $\ell(\theta)$ forms a line segment whose orientation is either $\theta$ or $\theta+\pi/2$.
Also, observe that the radius $r(\theta)$ of any possible outer square
and the line segment $\ell(\theta)$ are determined only by the bounding box $B(\theta)$
and its four contact points.
More specifically, $r(\theta)$ is of the form $b \cdot \sin(\theta + c)$ 
for some constants $b,c\in\Real$, so is a sinusoidal function of period $2\pi$.

Let $\rho(\theta)$ be the radius of a largest empty square among those
whose orientation is $\theta$ and whose center lies on $\ell(\theta)$.
Also, let $w(\theta)$ be the minimum possible width of 
a square annulus in orientation $\theta$ that encloses $P$
Then, by the above observation,
we have $w(\theta) = r(\theta) - \rho(\theta)$.
The problem of minimum width is to minimize $w(\theta)$ over all $\theta \in \OS$.

Now, consider each boxed MES class $\Mes_i\in\bMES_E$ and define
$\rho_{\Mes_i}(\theta)$ to be the radius of a largest square
among those contained in $\Mes_i(\theta)$
and centered at some point on $\ell(\theta)$.
Since our MES classes describe all possible maximal empty squares,
it is obvious that
\[ \rho(\theta) = \max_{\Mes_i\in\bMES} \rho_{\Mes_i}(\theta),\]
and thus
 \[ \min_{\theta \in \OS} w(\theta) 
  = \min_{\theta \in\OS} \min_{\Mes_i\in\bMES} (r(\theta) - \rho_{\Mes_i}(\theta))
  = \min_{\Mes_i\in\bMES} \min_{\theta\in\OS} (r(\theta) - \rho_{\Mes_i}(\theta)),
\]
where we define $\rho_{\Mes_i}(\theta) = 0$ for $\theta$ out of the interval of $\Mes_i$.
A similar analysis to that in Bae~\cite{b-marsap-19X} proves that $\delta_{\Mes_i}(\theta)$
is piecewise sinusoidal of period $2\pi$ with $O(1)$ breakpoints.
Hence, we can find the minimum $\min_{\theta\in\OS} (r(\theta) - \rho_{\Mes_i}(\theta))$
in $O(1)$ time for each $\Mes_i\in\bMES$.
This results in an $O(s_4\log n+ n^2)$-time algorithm that computes a minimum-width square annulus
in arbitrary orientation.
A square annulus of minimum area can also be 
computed similarly in the same time bound.

\begin{theorem} \label{thm:sq_annulus}
 Given a set $P$ of $n$ points in the general position,
 a square annulus of minimum width or minimum area in arbitrary orientation
 that encloses $P$ can be computed in $O(n^2 \log n)$ time.
\end{theorem}




%

{
\bibliographystyle{abbrv}
\bibliography{emptysquares}
}

%

\end{document}